\documentclass[a4paper,UKenglish,cleveref, autoref, thm-restate]{lipics-v2021}

\newcommand*{\fullpaper}{}

\ifdefined\fullpaper
\pdfoutput=1 
\fi
\ifdefined\fullpaper
\hideLIPIcs  
\fi

\title{Compressed Index with Construction in Compressed Space}

\author{Dmitry Kosolobov}{St.~Petersburg State University, Saint Petersburg, Russia}{dkosolobov@mail.ru}{https://orcid.org/0000-0002-2909-2952}{Supported by the Russian Science Foundation (RSF), project 24-71-00062}

\authorrunning{D. Kosolobov} 

\Copyright{Dmitry Kosolobov} 

\ccsdesc[100]{Theory of computation~Design and analysis of algorithms} 

\keywords{compressed index, pattern matching, string complexity, grammar, block tree} 

\category{} 

\ifdefined\fullpaper
\else
\relatedversion{} 
\relatedversiondetails{Full Version}{https://arxiv.org/abs/2602.13735} 
\fi

\acknowledgements{}

\nolinenumbers 

\EventEditors{Philip Bille and Nicola Prezza}
\EventNoEds{2}
\EventLongTitle{37th Annual Symposium on Combinatorial Pattern Matching (CPM 2026)}
\EventShortTitle{CPM 2026}
\EventAcronym{CPM}
\EventYear{2026}
\EventDate{June 15--17, 2026}
\EventLocation{Copenhagen, Denmark}
\EventLogo{}
\SeriesVolume{369}
\ArticleNo{34}

\usepackage{tikz}

\makeatother

\newcommand{\lvec}[1]{\overset{{}_{\leftarrow}}{#1}}
\newcommand{\lrange}[1]{\overleftarrow{#1}}

\newcommand\dd{\,..\,}
\newcommand\id{\mathsf{id}}
\newcommand\lbit{\mathsf{bit}}
\newcommand\vbit{\mathsf{vbit}}
\newcommand\rleft{\mathop{\mathsf{left}}}
\newcommand\rright{\mathop{\mathsf{right}}}
\newcommand{\sbeg}{\mathop{\mathsf{b}}}
\newcommand{\send}{\mathop{\mathsf{e}}}
\newcommand{\nil}{\bot}

\begin{document}

\sloppy

\maketitle

\begin{abstract}
	Suppose that we are given a string $s$ of length $n$ over an alphabet $\{0,1,\ldots,n^{O(1)}\}$ and $\delta$ is the string complexity of $s$, a known compression measure. We describe an index on $s$ with $O(\delta\log\frac{n}{\delta})$ space, measured in $O(\log n)$-bit machine words, which can search in $s$ any string of length $m$ in $O(m + (\mathrm{occ} + 1)\log^\epsilon n)$ time, where $\mathrm{occ}$ is the number of occurrences and $\epsilon > 0$ is any fixed constant (the big-O in the space bound hides factor $\frac{1}{\epsilon}$). Crucially, the index can be built in $O(n\log n)$ expected time by one left-to-right pass on the string $s$ in a streaming fashion with $O(\delta\log\frac{n}{\delta})$ construction space. The index does not use the Karp--Rabin fingerprints, and the randomization in the construction time can be eliminated by using deterministic dictionaries instead of hash tables (with a slowdown). The search time matches currently best results and the space is almost optimal (the known optimum is $O(\delta\log\frac{n}{\delta\alpha})$, where $\alpha = \log_\sigma n$ and $\sigma$ is the alphabet size, and it coincides with $O(\delta\log\frac{n}{\delta})$ when $\delta = O(n / \alpha^2)$). This is the first index that can be constructed within such space and with such time guarantees. To avoid uninteresting marginal cases, all above bounds are stated for $\delta \ge \Omega(\log\log n)$.
\end{abstract}

\newpage

\section{Introduction}

Given a string $s$ of length $n$ over the alphabet $\{0,1,\ldots,\sigma-1\}$ with $\sigma \le n^{O(1)}$, a compressed index is a data structure that stores $s$ in a compressed form supporting fast substring search. Currently best indexes use $O(\delta \log\frac{n}{\delta\alpha})$ space \cite{KociumakaNavarroPrezza,KociumakaNavarroPrezza2,KociumakaNavarroOlivares,KociumakaNavarroOlivares2}, measured in $O(\log n)$-bit machine words, where $\alpha = \log_\sigma n$ and $\delta$ is a compression measure called string complexity \cite{RaskhodnikovaEtAl}, and they can search a substring of length $m$ in $O(m + (\mathrm{occ} + 1)\log^\epsilon n)$ time, where $\mathrm{occ}$ is the number of occurrences and $\epsilon > 0$ is any fixed constant (the big-O in the space bound hides factor $\frac{1}{\epsilon}$). This space $O(\delta \log\frac{n}{\delta\alpha})$ is known to be tight for a wide range of parameters~\cite{KociumakaNavarroOlivares,KociumakaNavarroOlivares2} and it lowerbounds the space of most dictionary-based compressed indexes such as indexes on the Lempel--Ziv parsing~\cite{ChristiansenEttienne,LZ77}, run-length grammars~\cite{ChristiansenEtAl,KiefferYang}, Burrows--Wheeler transform~\cite{GagieNavarroPrezza}, and string attractors~\cite{KempaPrezza} (see \cite{Navarro3} for an overview).

There are essentially two primary methods to achieve $O(m \log n)$ search time within the space $O(\delta\log\frac{n}{\delta})$, which both enhance the basic idea of Claude and Navarro \cite{ClaudeNavarro}: (1)~a neat scheme \cite{GagieNavarroPrezza} first outlined by Gagie et al.~\cite{GGKNP2} that combines the Karp--Rabin hashing~\cite{DietzfelbingerEtAl,KarpRabin} and z-fast tries~\cite{BelazzouguiBoldiPaghVigna,BelazzouguiBoldiVigna}; (2) a variant of the so-called locally consistent parsing \cite{ChristiansenEtAl,ColeVishkin,Jez,MelhornSundarUhrig}. 
Known indexes that achieve $O(m + \log^\epsilon n)$ search time combine both approaches \cite{ChristiansenEtAl,KociumakaNavarroOlivares2}. Method (1) is applicable for many indexes; unfortunately, it has a serious drawback: the Karp--Rabin hashing must be collision-free for all substrings of length $2^k$, for $k = 1,2,\ldots$, and the only known algorithm that verifies this condition uses $O(n)$ space and $O(n\log n)$ time \cite{BilleEtAl2}, which is prohibitive in use cases where the uncompressed data barely fits into memory. This issue hinders a wider usage of such indexes and it is the reason for the scarcity of experimental studies that try to translate the theoretical state-of-the-art methods to practice (we note that there are many works on construction of indexes, e.g., \cite{AdudodlaKempa,ConradWilson,KempaLangmead,KopplKurpiczMeyer,TakabatakeISakamoto,ZhangEtAl}, but they never use the collision-free Karp--Rabin hashing as the best theoretical results do). 

We make a step forward to solve this problem by presenting an index with $O(\delta\log\frac{n}{\delta})$ space and $O(m + (\mathrm{occ} + 1)\log^\epsilon n)$ search time that can be built in $O(n\log n)$ expected time by one left-to-right pass in a streaming fashion with $O(\delta\log\frac{n}{\delta})$ construction space. We do not use the Karp--Rabin hashes and the randomization in the construction can be removed using deterministic dictionaries instead of hash tables (with a slowdown by an $O(\sqrt{\frac{\log n}{\log\log n}})$ factor, see \cite{AnderssonThorup,Ruzic}). The result is purely theoretical due to large constants under the big-O but, we believe, its ideas might inspire practical designs. We conjecture that the index can be further improved to $O(\delta\log\frac{n}{\delta\alpha})$ space (also the construction space) and $O(m/\alpha + (\mathrm{occ} + 1)\log^\epsilon n)$ time in the string packed model (see \cite{BilleGortzSkjoldjensen,MunroNavarroNekrich}); this direction is left for a future work.

\medskip

Throughout, $s$ is the input string of length $n$ with letters $s[0],s[1], \ldots,s[n{-}1]$ from an alphabet $\{0,1,\ldots,n^{O(1)}\}$. For $i,j$, let $s[i\dd j]$ be the \emph{substring} $s[i]s[i{+}1]\cdots s[j]$ (empty if $i > j$). For sequence $t = t_0, t_1,\ldots, t_{m-1}$ (maybe a string), denote by $\lvec{t} = t_{m-1}, \ldots, t_1, t_0$ its \emph{reversal}. 
Denote $|t| = m$, $s[i\dd j) = s[i\dd j{-}1]$, $[i\dd j] = \{i,i+1,\ldots,j\}$, $[i\dd j) = [i\dd j{-}1]$.

The string complexity, $\delta$, for $s$ is defined as $\delta = \max\{d_k(s) / k \colon k \in [1\dd n]\}$, where $d_k(s)$ denotes the number of distinct substrings of $s$ of length $k$ \cite{KociumakaNavarroPrezza2,RaskhodnikovaEtAl}.

We heavily use the concept of \emph{blocks}: a block $B$ for string $s$ is an object attached to a certain substring $s[i\dd j]$; denote $\sbeg(B) = i$ and $\send(B) = j$. The blocks are not ``fragments'': two distinct blocks $B$ and $B'$ can be attached to the same substring $s[i\dd j]$. If unambiguous, we treat a block $B$ as a string, writing $|B|$, $B = s[i'\dd j']$ (with not necessarily $i' = \sbeg(B)$), etc. Two blocks $B$ and $B'$ are \emph{equal as strings} if $s[\sbeg(B)\dd \send(B)] = s[\sbeg(B')\dd \send(B')]$. When defining a new block, we usually use words like ``generate'', ``create'', ``copy'' (when the new block is a copy of an old one) and it must be clear to which substring the new block is attached. We sometimes treat the concept of blocks loosely but, hopefully, it will be clear from the context.

\section{Main Ideas}
\label{sec:idea}

What follows is a high level description of main ideas behind our index. Details and precise definitions follow. We first focus on a static index. Its construction algorithm is in Section~\ref{sec:construction}.

Recent results on block trees \cite{KociumakaNavarroPrezza,NavarroPrezza} showed that the space $O(\delta \log\frac{n}{\delta})$ for indexes is surprisingly easy to obtain by, roughly, the following scheme (or its variants): one constructs a hierarchy of blocks with $\log n$ levels such that the level-$k$ blocks are $\lceil n/2^k\rceil$ substrings $s[i\cdot 2^k\dd (i+1)2^k)$ of length $2^k$ (the rightmost block might be shorter) and, for $k > 0$, each level-$k$ block has two child blocks from level $k-1$: its prefix and suffix of length $2^{k-1}$, respectively; then, each block $B = s[i\cdot 2^k\dd (i+1)2^k)$ for which the substring $\bar{B} = s[(i-1)2^k\dd (i+2)2^{k})$ has an occurrence at a position smaller than $(i-1)2^k$ is encoded by a pointer to the leftmost occurrence of $B$ in $s$ and all descendants of $B$ in the hierarchy are removed (for precise and detailed definitions, see \cite{BCGGKNOPT,BGGKOPT,KociumakaNavarroPrezza,Navarro2,NavarroPrezza}). The result is a block tree~\cite{BCGGKNOPT,Navarro2}

After a closer inspection of the proof from~\cite{KociumakaNavarroPrezza2} of the bound $O(\delta\log\frac{n}{\delta})$ for the block tree size, one can notice that the bound still holds if the block boundaries are slightly ``jiggled''. Our idea is to construct a ``jiggy block tree'', in which level-$k$ blocks are not of exact size $2^k$ but somewhat close to this ``on average'', and each block might have up to $O(\log\log n)$ children. The variability of block lengths will be used to form a kind of locally consistent parsing, more specifically, we will use a ``cut'' version of the so-called deterministic coin tossing (DCT) \cite{ColeVishkin} and the blocks will be built level by level bottom-up using this ``cut DCT'' (the idea of such level-by-level construction itself is not novel \cite{NishimotoEtAl} but it is usually used with the full DCT and it leads to a grammar of size $\Omega(\delta\log\frac{n}{\delta}\log^* n)$, not to a kind of block tree). Then, the same arguments as in \cite{KociumakaNavarroPrezza2} show that the resulting tree has $O(\delta\log\frac{n}{\delta})$ internal nodes. However, some nodes might have up to $\Theta(\log\log n)$ leaf children, blowing up the space by an $\Omega(\log\log n)$ factor, which seems unavoidable when using the DCT \cite{BirenzwigeEtAl,KosolobovSivukhin,MelhornSundarUhrig,NishimotoEtAl,SahinalpVishkin,TakabatakeISakamoto}. 

We eliminate the factor $\Omega(\log\log n)$ by introducing a novel intermediate step in the level by level bottom-up construction: given a block $B$ with $O(\log\log n)$ children $B_i,\ldots,B_j$, we greedily unite the children from left to right into larger ``intermediate blocks'' $B' = B_m\cdots B_{m+r-1}$ such that the sequence of blocks $B_m\cdots B_{m+r-1}$ occurs somewhere to the left in the block hierarchy; then, we store in the intermediate block $B'$ a pointer to this earlier occurrence. Due to the property of ``local consistency'' of DCT, the pointers will be somewhat ``aligned'' on the block structure and our jiggly block tree will resemble a grammar with elements of the block tree structure (like a collage system \cite{KidaEtAl} but more structured). There are many details in this scheme to make it efficient that are explained in the sequel. 

For the pattern matching, we follow the standard approach based on the locally consistent parsing and z-fast tries \cite{ChristiansenEtAl,KociumakaNavarroOlivares2} but we replace the Karp--Rabin fingerprints \cite{KarpRabin} with new deterministic fingerprints that are derived from our locally consistent block structure.

\section{Hierarchy of Blocks}
\label{sec:blocks}

At the core of our index is a \emph{hierarchy of blocks} (see Fig.~\ref{fig:blocks}, the leftmost picture): it has $O(\log n)$ levels where the topmost level has one block equal to $s$ and the bottom level $0$ contains $n$ one-letter blocks $s[0],\ldots,s[n{-}1]$; the concatenation of all blocks from one level equals $s$; for $\ell > 0$, each level-$\ell$ block $B$ either is a copy of a level-$(\ell{-}1)$ block $B'$ with $\sbeg(B) = \sbeg(B')$ and $\send(B) = \send(B')$ or is the concatenation of corresponding consecutive level-$(\ell{-}1)$ blocks.
We define the hierarchy of blocks by a bottom-up construction process. 

The bottom level-$0$ blocks are $n$ one-letter blocks $s[0], s[1], \ldots, s[n{-}1]$.
On a generic step with $k \ge 0$, we have a sequence of blocks $B_1, \ldots, B_b$ of level $2k$ whose concatenation is $s$ itself. To produce new blocks of levels $2k+1$ and $2k + 2$, the construction process will unite some adjacent blocks or will copy some blocks unaffected; in particular, all blocks of length ${>}2^k$ are copied to levels $2k + 1$ and $2k + 2$ unaffected. Every block $B$ is assigned an $O(\log n)$-bit integer identifier $\id(B)$; the level-$0$ blocks have identifiers equal to the letters they represent (recall that the alphabet is $\{0,1,\ldots,n^{O(1)}\}$). Blocks with equal identifiers are equal as strings but the converse is not necessarily true. We choose an arbitrary unambiguous method to encode any pair of $O(\log n)$-bit integers $x,y$ into one $O(\log n)$-bit integer, denoted $\langle x,y\rangle$. The process maintains a counter to assign new unique identifiers to new blocks.

\textbf{\boldmath Level $2k+1$.} 
In the level-$2k$ blocks $B_1,\ldots, B_b$, we identify all maximal ``runs'' of blocks $B_{i}, B_{i+1}, \ldots, B_{i+r-1}$ with equal identifiers such that $|B_{j}| \le 2^k$, for $j \in [i\dd i{+}r)$, and $i = 1$ or $\id(B_{i-1}) \ne \id(B_{i})$, and $i+r-1 = b$ or $\id(B_{i+r}) \ne \id(B_{i+r-1})$. Each such run for which $r > 1$ is substituted by a new block $B = B_{i} B_{i+1} \cdots B_{i+r-1}$ whose identifier encodes the pair $\langle\id(B_{i}), r\rangle$. Hence, any two distinct runs of the same length $r > 1$ with the same identifiers $\id(B_{i})$ of the run blocks are substituted by blocks with the same identifier $\langle\id(B_{i}), r\rangle$.  Thus, the blocks for level $2k+1$ consist of copies of blocks from level $2k$ with length ${>}2^k$, copies of blocks that formed runs of length one, and the blocks substituted in place of longer runs. 

\textbf{\boldmath Level $2k+2$.}
For any distinct integers $x, y \ge 0$, denote by $\lbit(x,y)$ the index ($0$-based) of the lowest bit in which the bit representations of $x$ and $y$ differ; e.g., $\lbit(8,0) = 3$. Denote $\vbit(x,y) = 2\cdot\lbit(x,y) + a$, where $a$ is the bit of $x$ with index $\lbit(x,y)$. Cole and Vishkin based their deterministic coin tossing (DCT) technique on the following observation.

\begin{lemma}[{see \cite{ColeVishkin,KosolobovSivukhin}}]
	Given a string $a_0 a_1 \cdots a_m$ over an alphabet $[0\dd 2^u)$ such that $a_{i-1} \ne a_{i}$ for any $i \in [1\dd m]$, the string $b_1 b_2 \cdots b_{m}$ such that $b_i = \vbit(a_{i-1}, a_{i})$, for $i \in [1\dd m]$, satisfies $b_{i-1} \ne b_{i}$, for any $i \in [2\dd m]$, and $b_i \in [0\dd 2u)$, for any $i \in [1\dd m]$.\label{lem:vishkin}
\end{lemma}

Let $B_1, \ldots, B_b$ be all blocks on level $2k+1$. For each $B_i$, we assign two identifiers, denoted $\id'(B_i)$ and $\id''(B_i)$. If $|B_i| > 2^k$, define $\id'(B_i) = \infty$. If $|B_i| \le 2^k$, then $\id(B_{i-1}) \ne \id(B_i)$ since the previous stage united all runs; define $\id'(B_i) = \vbit(\id(B_{i-1}), \id(B_i))$ if $i > 1$ and $|B_{i-1}| \le 2^k$, and $\id'(B_i) = \infty$ otherwise. Define $\id''(B_i) = \vbit(\id'(B_{i-1}), \id'(B_i))$ if $i > 1$ and $\id'(B_{i-1}) \ne \infty$ and  $\id'(B_{i}) \ne \infty$, and define $\id''(B_i) = \infty$ otherwise. This computation of $\id''$ corresponds to two first steps of the standard DCT, i.e., it is a ``cut'' variant of DCT.

For any maximal range of blocks $B_p,\ldots,B_q$ such that $|B_h|\le 2^k$ for each $h \in [p\dd q]$ (i.e., $p = 1$ or $|B_{p-1}| > 2^k$, and $q = b$ or $|B_{q+1}| > 2^k$), we have $\id''(B_p) = \id''(B_{p+1}) = \infty$ (or just $\id''(B_p) = \infty$ if $p = q$) and, by Lemma~\ref{lem:vishkin},  for each $h \in [p{+}2\dd q]$, we have $\id''(B_{h-1}) \ne \id''(B_h)$ and  $0\le \id''(B_h) \le O(\log\log n)$  (here we used the assumption $\id(B_h) \le n^{O(1)}$). We will split $B_p,\ldots,B_q$ into ranges ending at points corresponding to local minima of $\id''$. To this end, we mark each block $B_i$, with $i \in [1\dd b]$, for which one of the following conditions holds:
\begin{itemize} 
	\item[(a)] $|B_i| > 2^k$ or $|B_{i+1}| > 2^k$ or $i = b$ (the last block on level $2k+1$); 
	\item[(b)] $i > 2$ and $\infty > \id''(B_{i-2}) > \id''(B_{i-1})$ and $\id''(B_{i-1}) < \id''(B_i) < \infty$ 
	(i.e., $\id''(B_i)$ is immediately preceded by a local minimum).
\end{itemize} 
The marking splits all blocks $B_1,\ldots,B_b$ on level $2k+1$ into disjoint ranges as follows: each range $B_i,\ldots,B_j$ consists of unmarked blocks $B_i,\ldots,B_{j-1}$ and the marked block $B_j$ and either $i = 1$ or $B_{i-1}$ is marked. 
For each such range $B_i,\ldots,B_j$: if $i = j$, the process copies $B_i$ to level $2k + 2$; if $i < j$, it creates a new block $B = B_i\cdots B_j$ for level $2k+2$ assigning to it a new identifier unless another block $B'$ produced by the same sequence $\id(B_i),\ldots, \id(B_j)$ has already appeared earlier in the process (on any level), in which case we assign $\id(B) = \id(B')$. 
Since $\id'' \le O(\log\log n)$, local minima for $\id''$ occur often and, thus, $j - i \le O(\log\log n)$.

\textbf{Properties of blocks.}
The main properties of the hierarchy are its ``local consistency'' and ``local sparsity''. The proof of the former is standard \cite{KociumakaNavarroPrezza2} and it was moved to Appendix~\ref{appx:hierarchy}.

\begin{restatable}[local consistency]{lemma}{consistency}
	Fix $\ell = 2k$ or $\ell = 2k+1$. Let ${B}_1,{B}_2,\ldots,B_b$ denote all level-$\ell$ blocks. For any $i$, $j$ and block ${B}_h$, if $s[i{-}2^{k+4}\dd j{+}2^{k+4}] = s[\sbeg({B}_h){-}2^{k+4}\dd \send({B}_h){+}2^{k+4}]$, then there exists a level-$\ell$ block ${B}_{h'}$ such that $\id({B}_{h'}) = \id({B}_h)$, $\sbeg({B}_{h'}) = i$, and  $\send({B}_{h'}) = j$.\label{lem:consistency}
\end{restatable}

We call a position $h \in [0\dd n)$ a \emph{block boundary} for level $\ell$ if $h = \send(B)$ for some level-$\ell$ block $B$.
The ``local sparsity'' formalizes the intuition that blocks on levels $2k$ and $2k + 1$ have length ${\ge}2^k$ ``on average''. The proof is quite technical and it mostly repeats the arguments from the proof of \cite[Lemma~11]{KosolobovSivukhin} word by word. The proof  \ifdefined\fullpaper has been moved to Appendix~\ref{appx:hierarchy}.\else can be found in the full version of the paper \cite[Appendix~A]{KosolobovFull}\fi

\begin{restatable}[local sparsity]{lemma}{sparsity}
	For any $i < j$, the number of block boundaries in $[i\dd j)$ on levels $2k$ and $2k + 1$ is at most $2^6\lceil\frac{j-i}{2^k}\rceil$. In particular, there are $O(n/2^k)$ blocks on these levels.
	\label{lem:sparsity}
\end{restatable}

Lemma~\ref{lem:sparsity} implies that there are $O(\log n)$ levels in the hierarchy. We conclude the section with a technical lemma. Its proof is quite straightforward and has been moved to Appendix~\ref{appx:hierarchy}.

\begin{restatable}{lemma}{uniformity}
	Any blocks $B$ and $B'$ with $\id(B) = \id(B')$ are equal as strings and both were created on the same level.
	\label{lem:uniformity}
\end{restatable}

\begin{figure}[hbt]
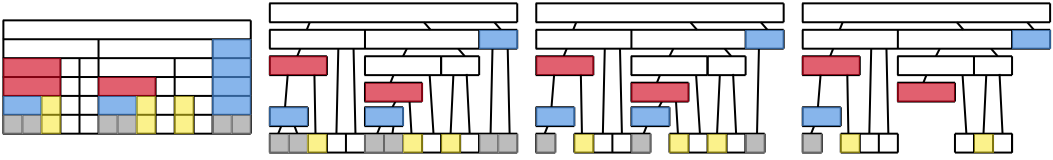

	\tikzset{every picture/.style={line width=0.75pt}} 
	

	
	\caption{A schematic depiction of the hierarchy of blocks $H$. Colored rectangles with the same color depict blocks with the same $\id$; white rectangles are blocks with any $\id$. From left to right: the hierarchy $H$, $H$ with rules 1--2 applied, $H$ with rule 4 applied, $H$ with rule 3 applied (i.e., it is $\hat{H}$).}
	\label{fig:blocks}
\end{figure}

\section{Jiggly Block Tree}
\label{sec:jiggly}

The described hierarchy of blocks for $s$ naturally forms a tree $H$ with blocks as vertices in which the root is the (only) topmost block and the parent of any level-$\ell$ block $B$ with $B\ne s$ is the level-$(\ell+1)$ block $B'$ such that $[\sbeg(B') \dd \send(B')] \supseteq [\sbeg(B) \dd \send(B)]$ (note that $B$ and $B'$ may be equal as strings).
For our index, we create a tree $\hat{H}$ that is a copy of $H$ with vertices removed according to the following rules 1--4 while it is possible to apply them (see Fig.~\ref{fig:blocks}): 
\begin{enumerate}
	\item  if a level-$0$ block $B$ has parent $B'$ with $\id(B) = \id(B')$,  merge $B$ and $B'$ by removing $B'$ and connecting $B$ to the parent of $B'$;
	\item if, for $\ell > 0$, a level-$\ell$ block $B$ has parent $B'$ with $\id(B) = \id(B')$, merge $B$ and $B'$ by removing $B$ and connecting all children of $B$ to $B'$;
	\item if, for $\ell > 0$ and $\ell' > 0$, a level-$\ell$ block $B$ has a ``preceding'' level-$\ell'$ block $B'$ such that $\id(B) = \id(B')$ and $\sbeg(B') < \sbeg(B)$, remove all descendants of $B$;
	\item if, for $\ell > 0$, all children of a level-$\ell$ block $B$ have equal $\id$, remove them all except for the leftmost child (note that, by removing a child, we remove its descendants).
\end{enumerate}
The idea of such construction is not novel (however, usually the full DCT is used, not the ``cut'' DCT as we do) and it leads to an index with $\omega(\delta \log\frac{n}{\delta})$ space \cite{BirenzwigeEtAl,KosolobovSivukhin,MelhornSundarUhrig,NishimotoEtAl,SahinalpVishkin,TakabatakeISakamoto}. As it turns out, the obtained tree has $O(\delta \log\frac{n}{\delta})$ non-leaf vertices and the main issue with space is due to some vertices having up to $O(\log\log n)$ leaves from levels $2k + 1$. To reduce the space, our idea is to parse the sequences of leaves by a specialized variant of LZ77 \cite{LZ77}. The construction is quite technical and requires some preparations to define it.

Let $B_1,\ldots,B_b$ be all blocks on a level $2k + 1$. Recall that the block construction process has marked some of the blocks, thereby splitting them into ranges $B_i,\ldots,B_j$, where $B_i,\ldots,B_{j-1}$ are unmarked and $B_j$ is marked. Fix such range $B_i,\ldots,B_j$ with $i < j$ such that $B_i,\ldots,B_j$ were not removed from $\hat{H}$ after the rules 1--4 above. Let $B$ be the parent of $B_i,\ldots,B_j$. We are to parse $B_i,\ldots,B_j$ into subranges. To define the parsing, we need some notation. 

For any block $B_{m}$ and $\ell \ge 0$, choose the largest $h \le \ell$ such that, for all $h'\in [m{-}h\dd m)$, $B_{h'}$ is not marked; define $\rleft(m,\ell)$ as the sequence $\beta,\id(B_{m-h}), \id(B_{m-h+1}), \ldots,\id(B_{m{-}1})$, where $\beta$ is empty (i.e., it is really not in the sequence) if $h = \ell$, and $\beta$ is a special symbol $\$$ if $h < \ell$. Choose the largest $h \le \ell$ such that, for all $h'\in [m\dd m{+}h{-}1)$, $B_{h'}$ is not marked (but $B_{m+h-1}$ may be marked); define $\rright(m,\ell)$ as the sequence $\id(B_{m}),\id(B_{m+1}), \ldots, \id(B_{m+h-1})$.

We greedily parse $B_i,\ldots,B_{j}$ into subranges from left to right: 
suppose that $B_i,\ldots,B_{m-1}$ have already been parsed (initially $m=i$), let $B_{m},\ldots,B_{m+r-1}$ be the longest subrange with $m+r-1 \le j$ for which there is $m' < m$ such that $\rleft(m',4r) = \rleft(m,4r)$ and $\rright(m',5r) = \rright(m,5r)$ (note that $\id(B_m),\ldots,\id(B_{m+r-1})$ must be a prefix of $\rright(m,5r)$); let $r = 1$ if there is no such $B_{m},\ldots,B_{m+r-1}$; once $r$ is determined, we have obtained a new subrange $B_{m},\ldots,B_{m+r-1}$ and we  continue the parsing for $B_{m+r},\ldots,B_j$, assigning $m := m+r$. 

Define a new tree $J$, called the \emph{\textbf{jiggly block tree}}, that is a copy of $\hat{H}$ with the following modifications. For each level $2k + 1$, we consider each range $B_i,\ldots,B_j$ as described above such that all $B_i,\ldots,B_j$ were not removed from $\hat{H}$ after the rules 1--4, and we parse  $B_i,\ldots,B_j$  into subranges as described above; for each subrange $B_m,\ldots,B_{m+r-1}$ with $r > 1$, we remove from $J$ the blocks $B_m,\ldots,B_{m+r-1}$ (which are children of $B$) and insert in their place a new block $B'$ with parent $B$ and no children; $B'$ is called \emph{intermediate}. We do not assign $\id$ to $B'$. 

By a simple analysis, one can show that any ``original'' block $B$ of $H$ cannot disappear from $J$, always a copy of $B$ remains, which is formalized as follows (the proof is in Appendix~\ref{appx:jiggly}). 

\begin{restatable}{lemma}{leftmost}
	For $d \ge 0$, let $B$ be the leftmost topmost block with $\id(B) = d$ and $|B| > 1$ in the hierarchy $H$, i.e., $\sbeg(B)$ is minimal for $B$ among all blocks with this $\id$ and $B$ is such block with maximal level. The jiggly block tree $J$ retains the block $B$ and $B$ has children in $J$.\label{lem:leftmost}
\end{restatable}

To support the basic navigation on the string $s$, the tree $J$ is equipped as follows.
Each non-intermediate block $B$ in $J$ stores $\sbeg(B)$, $\send(B)$, $\id(B)$, pointers to children of $B$, and the lowest level $\ell(B)$ on which a block with identifier $\id(B)$ was created in $H$. For each non-leaf block $B$ in $J$, two cases are possible: (1)~$B$ has children $B'_i,\ldots,B'_j$ in $J$ (some of which might be intermediate blocks inserted in place of old children of $B$ from $H$) whose concatenation, if viewed as strings, equals $s[\sbeg(B)\dd\send(B)]$ (i.e., no ``gaps'' between the children); (2)~$B$ has only one child $B'$ and all removed children had the same identifiers $\id(B')$, i.e., $B$ was a run block with exactly $|B| / |B'|$ children in $H$. If a block $B$ in $J$ with $|B| > 1$ has no children, then either $B$ is intermediate or, by Lemma~\ref{lem:leftmost}, there is a block $B'$ with children in $J$ such that $\id(B') = \id(B)$. In the latter case, we store in $B$ a pointer to $B'$; observe that if $B$ is from a level $\ell$, then all children $B'_i,\ldots,B'_j$ of $B'$ are from levels ${\le}\ell - 1$ since, by Lemma~\ref{lem:uniformity}, $B$ and $B'$ originate from blocks created on the same level in the hierarchy $H$. Hence, when one moves from the level-$\ell$ block $B$ to $B'$ and, then, to a child of $B'$, the level decreases to at most $\ell - 1$. 
Thus, if there are no intermediate blocks in $J$, one can read $s[\sbeg(B)\dd \send(B)]$ for any block $B$ in $J$ in $O(|B|)$ time by naively descending in $J$ as in grammar indexes \cite{ChristiansenEtAl}. 

Consider an intermediate block $B$ in $J$ from a level $2k + 1$. By construction, it has no children in $J$. Let $B'$ be its parent. Denote by $B_1,\ldots,B_b$ all blocks from level $2k + 1$ in the hierarchy $H$; we mark them according to conditions (a) and (b) from Section~\ref{sec:blocks} so that we can use the notation $\rleft$ and $\rright$ from now on. Suppose that $B' = B_i \cdots B_j$, where $B_i,\ldots,B_j$ is the range of blocks from which $B'$ was produced in $H$. Since $B$ is intermediate, we have $B = B_m \cdots B_{m+r-1}$, for some $m$ and $r$ such that $[m\dd m{+}r) \subseteq [i\dd j]$. By construction, there is $m' < m$ with $\rleft(m',4r) = \rleft(m,4r)$ and $\rright(m',5r) = \rright(m,5r)$. Choose the smallest $m''$ such that $\rright(m'',r) = \rright(m,r)$. Obviously, $m'' \le m'$. Let $B_{i'},\ldots,B_{j'}$ be a range of blocks in $H$ such that $i' \le m'' \le j'$, and let $\hat{B}$ be a block  on level $2k + 2$ produced from this range. Since $\rright(m'',r) = \rright(m,r)$, the definition of $\rright$ implies that the blocks $B_{m''},\ldots,B_{m''+r-2}$ are unmarked and, therefore, $[m''\dd m''{+}r) \subseteq [i'\dd j']$. Since $m''$ is the smallest such index, $\hat{B}$ must be the leftmost block in $H$ from level $2k+2$ produced from the sequence $\id(B_{i'}),\ldots,\id(B_{j'})$, which, by Lemma~\ref{lem:uniformity}, implies that $\hat{B}$ is the leftmost block with the identifier $\id(\hat{B})$ in the whole hierarchy $H$. By Lemma~\ref{lem:leftmost}, the block $\hat{B}$ was retained in $J$ and has children in $J$. We store in $B$ a pointer to $\hat{B}$ and the following numbers required to identify the location of the copy of $B = B_m\cdots B_{m+r-1}$ among the children of $\hat{B}$: $\mathsf{off}(B) = m'' - i'$ and $\mathsf{r}(B) = r$. (Note that $\hat{B}$ might coincide with $B'$, albeit it is unusual.)

Some subranges of blocks in the range $B_{i'},\ldots,B_{j'}$ could be replaced with intermediate blocks in $J$. We claim that any such intermediate block $\hat{B}'$ that replaced some of the blocks $B_{m''},\ldots,B_{m''+r-1}$ has $\mathsf{r}(\hat{B}') \le r/2$. Suppose, to the contrary, that $\hat{B}' = B_{\hat{m}}\cdots B_{\hat{m}+\hat{r}-1}$ is such that $\hat{r} = \mathsf{r}(\hat{B}') > r/2$ and $[\hat{m}\dd \hat{m}{+}\hat{r})\cap [m''\dd m''{+}r) \ne \emptyset$. Since $4\hat{r} > r$, the concatenation of the sequences $\rleft(\hat{m},4\hat{r})$ and $\rright(\hat{m},5\hat{r})$ contains the sequence $\id(B_{m''}),\ldots,\id(B_{m''+r-1})$. Hence, by construction, there is $m''' < m''$ such that $\rright(m''',r) = \rright(m'',r)$, which contradicts the minimality of $m''$. 
This claim implies that, for an intermediate block $B$ in $J$, one can compute all $r$ children of $B$ from the original hierarchy of blocks $H$ in $O(r)$ time by the simple traversal of pointers in $J$. Since, by Lemma~\ref{lem:sparsity} (local sparsity), any block $B$ has $O(\sum_{k=1}^\infty |B|/2^k) = O(|B|)$ descendants in the hierarchy $H$, we obtain the following result.

\begin{lemma}
	Given a block $B$ in the jiggly block tree $J$, one can obtain all $r$ children of $B$ from the original hierarchy of blocks $H$ in $O(r)$ time; further, one can compute all descendant blocks for $B$ in $H$ in $O(|B|)$ time.\label{lem:traversal}
\end{lemma}

\begin{restatable}{theorem}{spacetheorem}
	The jiggly block tree for string $s$ of length $n$ has size $O(\delta\log\frac{n}{\delta})$, where $\delta$ is the string complexity of $s$, provided $\delta \ge \Omega(\log\log n)$.
	\label{thm:space}
\end{restatable}
\begin{proof}[Proof (sketch).]
	We are to estimate the number of blocks in $J$ as $O(\delta\log\frac{n}{\delta})$. The basic argument is as in \cite{KociumakaNavarroPrezza2}: we will estimate the number of internal blocks on each level $\ell$ as $O(\delta)$, which implies that the number of internal blocks on $2\log\frac{n}{\delta}$ lowest levels is $O(\delta\log\frac{n}{\delta})$, as required; then, by Lemma~\ref{lem:sparsity} (local sparsity), the number of blocks on levels $\ell \ge 2\log\frac{n}{\delta}$ is at most $O(\sum_{k=\log(n/\delta)}^{\infty} n / 2^k) = O(\delta)$. This scheme will be complicated by the subsequent counting of leaves since in $J$ a block may have up to $O(\log\log n)$ leaf-children.
	
	Fix $\ell = 2k + 1$ or $2k + 2$, for $k \ge 0$. Every non-leaf level-$\ell$ block $B$ in $J$ satisfies $|B| \le 2^k$ and its corresponding substring $s[\sbeg(B){-}2^{k+4}\dd \send(B){+}2^{k+4}]$ cannot occur at smaller positions since otherwise, due to Lemma~\ref{lem:consistency} (local consistency), there is a copy of  $B$ to the left of $B$ and, thus, all children of $B$ should be removed and a pointer to this copy should be stored in $B$. We associate with every such $B$ a group of $\Theta(2^k)$ substrings of length $2^{k+6}$, each of which covers $s[\sbeg(B){-}2^{k+4}\dd \send(B){+}2^{k+4}]$ and, hence, cannot occur at smaller positions, i.e., they are ``leftmost substrings''. Due to  Lemma~\ref{lem:sparsity} (local sparsity), at most $O(1)$ level-$\ell$ blocks can be associated with one substring of length $2^{k+6}$. Therefore, $\Theta(2^k)\cdot b \le O(d_{2^{k+6}})$, where $b$ is the number of non-leaf level-$\ell$ blocks $B$ and $d_{2^{k+6}}$ is the number of distinct substrings of length $2^{k+6}$ (which is equal to the number of leftmost substrings of length $2^{k+6}$). Thus, we obtain $b \le O(d_{2^{k+6}} / 2^{k+6}) \le O(\delta)$. (The described argument is exactly the same as in  \cite{KociumakaNavarroPrezza2}.)
	
	\begin{figure}[htb]
		\tikzset{every picture/.style={line width=0.75pt}} 
		
		\begin{tikzpicture}[x=0.75pt,y=0.75pt,yscale=-1,xscale=1]
			
			\draw    (19,80) -- (196.5,80) ;
			\draw    (19,70) -- (196.5,70) ;
			\draw    (213,70) -- (534.5,70) ;
			\draw    (49,70) -- (49,80) ;
			\draw    (39,70) -- (39,80) ;
			\draw    (59,70) -- (59,80) ;
			\draw    (69,70) -- (69,80) ;
			\draw    (83,70) -- (83,80) ;
			\draw    (90,70) -- (90,80) ;
			\draw    (100,70) -- (100,80) ;
			\draw    (120,70) -- (120,80) ;
			\draw    (130,70) -- (130,80) ;
			\draw    (140,70) -- (140,80) ;
			\draw    (154,70) -- (154,80) ;
			\draw    (164,70) -- (164,80) ;
			\draw    (184,70) -- (184,80) ;
			\draw    (192,70) -- (192,80) ;
			\draw    (31,70) -- (31,80) ;
			\draw    (299,70) -- (299,80) ;
			\draw    (289,70) -- (289,80) ;
			\draw    (309,70) -- (309,80) ;
			\draw    (319,70) -- (319,80) ;
			\draw    (333,70) -- (333,80) ;
			\draw    (340,70) -- (340,80) ;
			\draw    (350,70) -- (350,80) ;
			\draw    (370,70) -- (370,80) ;
			\draw    (380,70) -- (380,80) ;
			\draw    (390,70) -- (390,80) ;
			\draw    (404,70) -- (404,80) ;
			\draw    (414,70) -- (414,80) ;
			\draw    (424,70) -- (424,80) ;
			\draw    (436,70) -- (436,80) ;
			\draw    (275,70) -- (275,80) ;
			\draw  [fill={rgb, 255:red, 74; green, 144; blue, 226 }  ,fill opacity=0.49 ] (59,70) -- (164,70) -- (164,80) -- (59,80) -- cycle ;
			\draw  [fill={rgb, 255:red, 74; green, 144; blue, 226 }  ,fill opacity=0.49 ] (309,70) -- (414,70) -- (414,80) -- (309,80) -- cycle ;
			\draw  [fill={rgb, 255:red, 208; green, 2; blue, 27 }  ,fill opacity=0.47 ] (414,70) -- (436,70) -- (436,80) -- (414,80) -- cycle ;
			\draw  [fill={rgb, 255:red, 208; green, 2; blue, 27 }  ,fill opacity=0.47 ] (164,70) -- (192,70) -- (192,80) -- (164,80) -- cycle ;
			\draw    (233,70) .. controls (273,40) and (455.5,39) .. (494.5,69) ;
			\draw    (253,70) .. controls (293,40) and (475.5,39) .. (514.5,69) ;
			\draw    (273,70) .. controls (313,40) and (495.5,39) .. (534.5,69) ;
			\draw    (213,70) .. controls (253,40) and (435.5,39) .. (474.5,69) ;
			\draw  [fill={rgb, 255:red, 208; green, 2; blue, 27 }  ,fill opacity=0.47 ] (31,70) -- (59,70) -- (59,80) -- (31,80) -- cycle ;
			\draw  [fill={rgb, 255:red, 208; green, 2; blue, 27 }  ,fill opacity=0.47 ] (275,70) -- (309,70) -- (309,80) -- (275,80) -- cycle ;
			\draw    (213,80) -- (534.5,80) ;
			
			\draw (306,82) node [anchor=north west][inner sep=0.75pt]  [font=\footnotesize]  {$B_{m-4r}$};
			\draw (401,82) node [anchor=north west][inner sep=0.75pt]  [font=\footnotesize]  {$B_{m+5r-1}$};
			\draw (57,82) node [anchor=north west][inner sep=0.75pt]  [font=\footnotesize]  {$B_{m'-4r}$};
			\draw (152,82) node [anchor=north west][inner sep=0.75pt]  [font=\footnotesize]  {$B_{m'+5r-1}$};
			\draw (102,82) node [anchor=north west][inner sep=0.75pt]  [font=\footnotesize]  {$B_{m'}$};
			\draw (350,82) node [anchor=north west][inner sep=0.75pt]  [font=\footnotesize]  {$B_{m}$};
			\draw (358,27) node [anchor=north west][inner sep=0.75pt]  [font=\small]  {$\Theta ( 2^{k'})$};

		\end{tikzpicture}

		\caption{A schematic representation of a group of $\Theta(2^{k'})$ leftmost substrings of length $2^{k'}$ (depicted as arcs above) associated with a leaf block $B = B_m\cdots B_{m+r-1}$. The blue rectangles represent blocks $B_{m-4r},\ldots,B_{m+5r-1}$ and their copies to the left, $B_{m'-4r},\ldots,B_{m'+5r-1}$. The red rectangles represent that it is impossible to extend the subrange to $B_{m-4(r+1)},\ldots,B_{m+5(r+1)-1}$.}
		\label{fig:thera-2k}
	\end{figure}

	If every block in $J$ had $O(1)$ leaves, the total number of blocks would be $O($number of non-leaf blocks$)$. Unfortunately, some blocks may have up to $O(\log\log n)$ leaves. We analyze each level-$\ell$ leaf $B$ that is not the rightmost child of its parent $B'$. $B'$ was generated from a sequence of level-$\ell$ blocks: $B' = B_i\cdots B_j$ where a subrange $B_{m},\ldots,B_{m+r-1}$ generated the leaf block $B$ (intermediate, if $r > 1$, or not, if $r=1$). See Fig.~\ref{fig:thera-2k}. The subranges were produced by a greedy parsing that tried to produce the subrange $B_m,\ldots,B_{m+r}$ instead of $B_{m},\ldots,B_{m+r-1}$ but failed since there were no occurrences of the sequence $B_{m-4(r+1)},\ldots,B_{m+5(r+1)-1}$ (the subrange with a ``neighbourhood'' of $4(r+1)$ blocks to its left and right) to the left in the hierarchy of blocks. Therefore, the substring $s[a\dd b] = s[\sbeg(B_{m-4(r+1)}){-}2^{k+4} \dd \send(B_{m+5(r+1)-1}){+}2^{k+4}]$ cannot occur at smaller positions since, otherwise, by Lemma~\ref{lem:consistency} (local consistency), there would be an earlier occurrence for the sequence. We associate with $B$ a group of $\Theta(2^{k'})$ leftmost substrings of length $2^{k'+1}$ that cover $s[a\dd b]$, where $2^{k'}$ is the closest power of two to the length of $s[a\dd b]$. We consider separately each number $k'$ with all such groups of $\Theta(2^{k'})$ leftmost substrings, and the further analysis to bound the number of these groups by $O(\delta)$, for this $k'$, is similar as above, in principle, but it has many subtle details and special cases.
	
	All remaining details of the proof were moved to Appendix \ref{appx:jiggly}.
\end{proof}

\newcommand{\str}{\mathsf{str}}

\section{Substring Search}
\label{sec:search}

Let us augment the jiggly block tree $J$ to support searching in $s$. Our approach is standard but with modifications since we will not use Karp--Rabin fingerprints \cite{KarpRabin}. Instead, we invent somewhat analogous \emph{deterministic fingerprints}. Throughout, we use hash tables supporting queries in $O(1)$ time \cite{BFKK}. Let us fix a string  $t$ of length $m$, the pattern to search.

\textbf{\boldmath Parsing $t$.}
We create a compacted trie $T_{\id}$ and, for each non-run block $B$ in $J$ that was created from a range of blocks $B_i,\ldots,B_j$ from a level $2k + 1$ of the hierarchy $H$, we store in $T_{\id}$ the sequence $\id(B_i),\ldots,\id(B_j)$ (treated as a string with letters $\id(B_i),\ldots,\id(B_j)$).  
As is usual for compacted tries, the edge labels in $T_{\id}$ are sequences of $\id$s and there is a node $x$ such that $\str(x) =  \id(B_i),\ldots,\id(B_j)$, where $\str(x)$ denotes the sequence written on the root--$x$ path; we store in $x$ a pointer to $B$. The edges of $T_{\id}$ do not store their edge labels, only the lengths of these edge labels. Each internal node of $T_{\id}$ is augmented with a hash table that maps any block identifier to the child whose edge label starts with this identifier (if any). It follows from Theorem~\ref{thm:space} that the size of $T_{\id}$ is $O(\delta\log\frac{n}{\delta})$.

Our search algorithm constructs a hierarchy of blocks for $t$ analogous to $H$.  Level $0$ in the hierarchy for $t$ contains the one-letter blocks $t[0],\ldots,t[m{-}1]$ whose identifiers are the corresponding letters; on a generic step of the algorithm, for $\ell \ge 0$, we have all level-$\ell$ blocks $B_1,\ldots,B_b$ and $B_1\cdots B_b = t$. Suppose that $\ell = 2k$, for some $k\ge 0$. As in Section~\ref{sec:blocks}, we identify in $O(b)$ time all maximal runs of blocks $B_i,\ldots,B_{i+r-1}$ with equal identifiers (i.e., $\id(B_i) = \cdots = \id(B_{i+r-1})$) such that $|B_i| \le 2^k$. If $r = 1$ or $|B_i| > 2^k$, we copy the block $B_i$ to level $2k+1$. If $r > 1$, we substitute such run $B_i,\ldots,B_{i+r-1}$ with a new block $B$ for level $2k + 1$ with identifier $\langle\id(B_i),r\rangle$. Thus, all block for level $2k + 1$ are constructed. Suppose that $\ell = 2k + 1$, for some $k\ge 0$. As in Section~\ref{sec:blocks}, we assign in $O(b)$ time to each block $B_i$ an identifier $\id''(B_i)$ that is either $\infty$ (for example, if $|B_i| > 2^k$) or an integer that is at most $O(\log\log n)$; we then produce new blocks for level $2k + 2$ by identifying local minima for $\id''$ and values of $\id''$ equal to $\infty$; we omit details as they are the same as in Section~\ref{sec:blocks}. It remains to assign identifiers to new blocks. Suppose that a new such block $B$ was created from blocks $B_i,\ldots,B_j$ from level $2k + 1$. We descend from the root of the trie $T_{\id}$ reading the sequence $\id(B_i),\ldots,\id(B_j)$ and skipping sequences written on the edge labels; thus, we reach a node $x$ that stores a pointer to a block $B'$ in $J$ and we retrieve all skipped edge labels in $O(j - i + 1)$ time traversing $J$ from $B'$ as described in Lemma~\ref{lem:traversal}; we then assign $\id(B) = \id(B')$. If the described search has failed at any stage, we assign to $\id(B)$ a new unique identifier (for consistency, we also should maintain another trie analogous to $T_{\id}$ for new identifiers but, as it will be seen, this consistency is not very important). 

The algorithm processes all $b$ blocks on one level in $O(b)$ time. It then follows from Lemma~\ref{lem:sparsity} (local sparsity) that the overall time is $O(\sum_{k=0}^\infty m/2^k) = O(m)$.

\newcommand{\fin}{\mathsf{fin}}

\textbf{Deterministic fingerprints.} 
We view any substring $t[p\dd q)$ as the sequence of level-$0$ blocks $t[p], \ldots, t[q{-}1]$ and we define its deterministic fingerprint as $\fin(0, t[p],\ldots,t[q{-}1])$, where $\fin(\ell, B_1,\ldots,B_b)$ is the following recursive procedure that assigns a fingerprint to any contiguous subsequence of level-$\ell$ blocks $B_1,\ldots,B_b$ from the hierarchy of blocks of $t$.
\begin{enumerate} 
	\item The fingerprint $\fin(\ell, B_1,\ldots,B_b)$ of the empty sequence ($b = 0$) is the empty sequence.
	\item Case $\ell = 2k$, for $k \ge 0$. Choose maximal $i \ge 0$ such that $|B_i| \le 2^k$ and $\id(B_i) = \id(B_h)$, for all $h \in [1\dd i]$, or $i = b$. Choose minimal $j \le b$ such that $|B_{j+1}| \le 2^k$ and $\id(B_{j+1}) = \id(B_h)$, for all $h \in (j\dd b]$, or $j = 0$. If $i = b$, the fingerprint is $\langle\id(B_1),b\rangle$. If $i \ne b$, the fingerprint is the sequence $\langle\id(B_i),i\rangle, \fin(2k + 1, B'_1,\ldots, B'_{b'}),$ $ \langle\id(B_{j+1}),b-j\rangle$, where the element $\langle\id(B_i),i\rangle$ is missing if $i = 0$, $\langle\id(B_{j+1}),b-j\rangle$ is missing if $j=b$, and $B'_1,\ldots,B'_{b'}$ are all level-$(2k{+}1)$ blocks that are parents of the blocks $B_{i+1},\ldots,B_{j}$.
	\item Case $\ell = 2k + 1$, for $k \ge 0$. Redefine the identifiers $\id''$ for $B_1,\ldots,B_b$ as in Section~\ref{sec:blocks}. Choose the smallest $i$ such that either $|B_{i+1}| > 2^k$ (where $0 \le i < b$) or $\infty > \id''(B_{i-2}) > \id''(B_{i-1})$ and $\id''(B_{i-1}) < \id''(B_i) < \infty$ (where $3 \le i \le b$); let $i = b$ if there is no such $i$. Symmetrically, choose the largest $j$ such that either $|B_{j}| > 2^k$ or $\infty > \id''(B_{j-2}) > \id''(B_{j-1})$ and $\id''(B_{j-1}) < \id''(B_j) < \infty$ (where $3 \le j \le b$); let $j = b$ if there is no such $j$. Note that $i \le j$. The fingerprint is $\id(B_1),\ldots,\id(B_i),\fin(2k + 2, B'_1,\ldots, B'_{b'}), \id(B_{j+1}),\ldots,\id(B_b)$, where $B'_1,\ldots,B'_{b'}$ are all level-$(2k{+}2)$ blocks that are parents of the blocks $B_{i+1},\ldots,B_{j}$.
\end{enumerate}
The indices $i$ and $j$ in the procedure $\fin$ are chosen to be at block boundaries for the blocks $B'_1,\ldots,B'_{b'}$ from level $\ell + 1$ so that all children of $B'_1,\ldots,B'_{b'}$ are exactly the blocks $B_{i+1},\ldots,B_j$. Therefore, the concatenation of all blocks used in the fingerprint of $t[p\dd q)$ is equal to $t[p\dd q)$ (Fig.~\ref{fig:fingerprints}). Due to Lemma~\ref{lem:sparsity} (local sparsity), $\fin$ has at most $O(\log m)$ levels of recursion. Since $\id''$ is at most $O(\log\log n)$ whenever it is not $\infty$, it follows from the discussion of Section~\ref{sec:blocks} that case 3 in the recursion (when $\ell = 2k + 1$) uses at most $O(\log\log n)$ blocks for the resulting fingerprint. Hence, the size of the fingerprint is $O(\log m\log\log n)$.

\begin{figure}[ht]
	\tikzset{every picture/.style={line width=0.75pt}} 
	\begin{center}
		
	\begin{tikzpicture}[x=0.75pt,y=0.75pt,yscale=-1,xscale=1,scale=0.6]
		
		\draw  [draw opacity=0][fill={rgb, 255:red, 155; green, 155; blue, 155 }  ,fill opacity=0.72 ] (71.27,273.17) -- (201.27,273.17) -- (201.27,263.07) -- (71.27,263.07) -- cycle ;
		\draw  [draw opacity=0][fill={rgb, 255:red, 155; green, 155; blue, 155 }  ,fill opacity=0.72 ] (261.27,273.17) -- (391.27,273.17) -- (391.27,263.07) -- (261.27,263.07) -- cycle ;
		\draw  [draw opacity=0][fill={rgb, 255:red, 155; green, 155; blue, 155 }  ,fill opacity=0.72 ] (91.27,263.17) -- (191.27,263.17) -- (191.27,253.17) -- (91.27,253.17) -- cycle ;
		\draw  [draw opacity=0][fill={rgb, 255:red, 155; green, 155; blue, 155 }  ,fill opacity=0.72 ] (281.27,263.17) -- (381.27,263.17) -- (381.27,253.17) -- (281.27,253.17) -- cycle ;
		\draw  [draw opacity=0][fill={rgb, 255:red, 155; green, 155; blue, 155 }  ,fill opacity=0.72 ] (101.27,253.17) -- (171.27,253.17) -- (171.27,243.17) -- (101.27,243.17) -- cycle ;
		\draw  [draw opacity=0][fill={rgb, 255:red, 155; green, 155; blue, 155 }  ,fill opacity=0.72 ] (291.27,253.17) -- (361.27,253.17) -- (361.27,243.17) -- (291.27,243.17) -- cycle ;
		\draw  [draw opacity=0][fill={rgb, 255:red, 155; green, 155; blue, 155 }  ,fill opacity=0.72 ] (121.27,243.17) -- (151.27,243.17) -- (151.27,233.17) -- (121.27,233.17) -- cycle ;
		\draw  [draw opacity=0][fill={rgb, 255:red, 155; green, 155; blue, 155 }  ,fill opacity=0.72 ] (311.27,243.17) -- (341.27,243.17) -- (341.27,233.17) -- (311.27,233.17) -- cycle ;
		\draw  [draw opacity=0][fill={rgb, 255:red, 208; green, 2; blue, 27 }  ,fill opacity=0.48 ] (341.27,243.17) -- (411.27,243.17) -- (411.27,233.17) -- (341.27,233.17) -- cycle ;
		\draw  [draw opacity=0][fill={rgb, 255:red, 208; green, 2; blue, 27 }  ,fill opacity=0.48 ] (151.27,243.17) -- (221.27,243.17) -- (221.27,233.17) -- (151.27,233.17) -- cycle ;
		\draw  [draw opacity=0][fill={rgb, 255:red, 208; green, 2; blue, 27 }  ,fill opacity=0.48 ] (51.27,243.17) -- (121.27,243.17) -- (121.27,233.17) -- (51.27,233.17) -- cycle ;
		\draw  [draw opacity=0][fill={rgb, 255:red, 208; green, 2; blue, 27 }  ,fill opacity=0.48 ] (241.27,243.17) -- (311.27,243.17) -- (311.27,233.17) -- (241.27,233.17) -- cycle ;
		\draw  [draw opacity=0][fill={rgb, 255:red, 208; green, 2; blue, 27 }  ,fill opacity=0.48 ] (241.27,253.17) -- (291.27,253.17) -- (291.27,243.17) -- (241.27,243.17) -- cycle ;
		\draw  [draw opacity=0][fill={rgb, 255:red, 208; green, 2; blue, 27 }  ,fill opacity=0.48 ] (60.97,253.17) -- (101.27,253.17) -- (101.27,243.17) -- (60.97,243.17) -- cycle ;
		\draw  [draw opacity=0][fill={rgb, 255:red, 208; green, 2; blue, 27 }  ,fill opacity=0.48 ] (360.97,253.17) -- (411.27,253.17) -- (411.27,243.17) -- (360.97,243.17) -- cycle ;
		\draw  [draw opacity=0][fill={rgb, 255:red, 208; green, 2; blue, 27 }  ,fill opacity=0.48 ] (240.97,263.17) -- (281.27,263.17) -- (281.27,253.17) -- (240.97,253.17) -- cycle ;
		\draw  [draw opacity=0][fill={rgb, 255:red, 208; green, 2; blue, 27 }  ,fill opacity=0.48 ] (71.27,263.17) -- (91.27,263.17) -- (91.27,253.17) -- (71.27,253.17) -- cycle ;
		\draw  [draw opacity=0][fill={rgb, 255:red, 208; green, 2; blue, 27 }  ,fill opacity=0.48 ] (171.27,253.17) -- (191.27,253.17) -- (191.27,243.17) -- (171.27,243.17) -- cycle ;
		\draw  [draw opacity=0][fill={rgb, 255:red, 208; green, 2; blue, 27 }  ,fill opacity=0.48 ] (381.27,263.17) -- (401.27,263.17) -- (401.27,253.17) -- (381.27,253.17) -- cycle ;
		\draw  [draw opacity=0][fill={rgb, 255:red, 208; green, 2; blue, 27 }  ,fill opacity=0.48 ] (191.27,263.17) -- (201.27,263.17) -- (201.27,253.17) -- (191.27,253.17) -- cycle ;
		\draw    (71.27,263.17) -- (71.27,273.17) ;
		\draw    (50.97,273.17) -- (411.27,273.17) ;
		\draw    (50.97,263.17) -- (411.27,263.17) ;
		\draw    (81.27,263.17) -- (81.27,273.17) ;
		\draw    (91.27,263.17) -- (91.27,273.17) ;
		\draw    (101.27,263.17) -- (101.27,273.17) ;
		\draw    (111.27,263.17) -- (111.27,273.17) ;
		\draw    (121.27,263.17) -- (121.27,273.17) ;
		\draw    (131.27,263.17) -- (131.27,273.17) ;
		\draw    (141.27,263.17) -- (141.27,273.17) ;
		\draw    (151.27,263.17) -- (151.27,273.17) ;
		\draw    (161.27,263.17) -- (161.27,273.17) ;
		\draw    (171.27,263.17) -- (171.27,273.17) ;
		\draw    (181.27,263.17) -- (181.27,273.17) ;
		\draw    (191.27,263.17) -- (191.27,273.17) ;
		\draw    (201.27,263.17) -- (201.27,273.17) ;
		\draw    (261.27,263.17) -- (261.27,273.17) ;
		\draw    (271.27,263.17) -- (271.27,273.17) ;
		\draw    (281.27,263.17) -- (281.27,273.17) ;
		\draw    (291.27,263.17) -- (291.27,273.17) ;
		\draw    (301.27,263.17) -- (301.27,273.17) ;
		\draw    (311.27,263.17) -- (311.27,273.17) ;
		\draw    (321.27,263.17) -- (321.27,273.17) ;
		\draw    (331.27,263.17) -- (331.27,273.17) ;
		\draw    (341.27,263.17) -- (341.27,273.17) ;
		\draw    (351.27,263.17) -- (351.27,273.17) ;
		\draw    (361.27,263.17) -- (361.27,273.17) ;
		\draw    (371.27,263.17) -- (371.27,273.17) ;
		\draw    (381.27,263.17) -- (381.27,273.17) ;
		\draw    (391.27,263.17) -- (391.27,273.17) ;
		\draw    (51.27,263.17) -- (51.27,273.17) ;
		\draw    (61.27,263.17) -- (61.27,273.17) ;
		\draw    (71.27,263.17) -- (71.27,273.17) ;
		\draw    (81.27,263.17) -- (81.27,273.17) ;
		\draw    (241.27,263.17) -- (241.27,273.17) ;
		\draw    (251.27,263.17) -- (251.27,273.17) ;
		\draw    (261.27,263.17) -- (261.27,273.17) ;
		\draw    (271.27,263.17) -- (271.27,273.17) ;
		\draw    (381.27,263.17) -- (381.27,273.17) ;
		\draw    (391.27,263.17) -- (391.27,273.17) ;
		\draw    (401.27,263.17) -- (401.27,273.17) ;
		\draw    (411.27,263.17) -- (411.27,273.17) ;
		\draw    (91.27,253.17) -- (91.27,263.17) ;
		\draw    (101.27,253.17) -- (101.27,263.17) ;
		\draw    (111.27,253.17) -- (111.27,263.17) ;
		\draw    (121.27,253.17) -- (121.27,263.17) ;
		\draw    (131.27,253.17) -- (131.27,263.17) ;
		\draw    (141.27,253.17) -- (141.27,263.17) ;
		\draw    (151.27,253.17) -- (151.27,263.17) ;
		\draw    (161.27,253.17) -- (161.27,263.17) ;
		\draw    (171.27,253.17) -- (171.27,263.17) ;
		\draw    (191.27,253.17) -- (191.27,263.17) ;
		\draw    (281.27,253.17) -- (281.27,263.17) ;
		\draw    (291.27,253.17) -- (291.27,263.17) ;
		\draw    (301.27,253.17) -- (301.27,263.17) ;
		\draw    (311.27,253.17) -- (311.27,263.17) ;
		\draw    (321.27,253.17) -- (321.27,263.17) ;
		\draw    (331.27,253.17) -- (331.27,263.17) ;
		\draw    (341.27,253.17) -- (341.27,263.17) ;
		\draw    (351.27,253.17) -- (351.27,263.17) ;
		\draw    (361.27,253.17) -- (361.27,263.17) ;
		\draw    (381.27,253.17) -- (381.27,263.17) ;
		\draw    (71.27,253.17) -- (201.27,253.17) ;
		\draw    (241.27,253.17) -- (411.27,253.17) ;
		\draw    (71.27,253.17) -- (71.27,263.17) ;
		\draw    (241.27,253.17) -- (241.27,263.17) ;
		\draw    (401.27,253.17) -- (401.27,263.17) ;
		\draw    (211.27,263.17) -- (211.27,273.17) ;
		\draw    (201.27,253.17) -- (201.27,263.17) ;
		\draw    (90.97,243.17) -- (191.27,243.17) ;
		\draw    (280.97,243.17) -- (411.27,243.17) ;
		\draw    (241.27,243.17) -- (241.27,253.17) ;
		\draw    (291.27,243.17) -- (291.27,253.17) ;
		\draw    (311.27,243.17) -- (311.27,253.17) ;
		\draw    (341.27,243.17) -- (341.27,253.17) ;
		\draw    (361.27,243.17) -- (361.27,253.17) ;
		\draw    (411.27,243.17) -- (411.27,253.17) ;
		\draw    (191.27,243.17) -- (191.27,253.17) ;
		\draw    (171.27,243.17) -- (171.27,253.17) ;
		\draw    (151.27,243.17) -- (151.27,253.17) ;
		\draw    (121.27,243.17) -- (121.27,253.17) ;
		\draw    (101.27,243.17) -- (101.27,253.17) ;
		\draw    (61.27,243.17) -- (61.27,253.17) ;
		\draw    (51.27,253.17) -- (161.27,253.17) ;
		\draw    (60.97,243.17) -- (161.27,243.17) ;
		\draw    (221.27,243.17) -- (341.27,243.17) ;
		\draw    (311.27,233.17) -- (311.27,243.17) ;
		\draw    (341.27,233.17) -- (341.27,243.17) ;
		\draw    (151.27,233.17) -- (151.27,243.17) ;
		\draw    (121.27,233.17) -- (121.27,243.17) ;
		\draw    (51.27,233.17) -- (51.27,243.17) ;
		\draw    (241.27,233.17) -- (241.27,243.17) ;
		\draw    (411.27,233.17) -- (411.27,243.17) ;
		\draw    (51.27,243.17) -- (221.27,243.17) ;
		\draw    (110.97,233.17) -- (221.27,233.17) ;
		\draw    (50.97,233.17) -- (151.27,233.17) ;
		\draw    (221.27,233.17) -- (341.27,233.17) ;
		\draw    (310.97,233.17) -- (411.27,233.17) ;
		\draw    (221.27,233.17) -- (221.27,243.17) ;
		\draw    (241.27,253.17) -- (241.27,263.17) ;
		\draw    (231.27,263.17) -- (231.27,273.17) ;
		\draw    (221.27,263.17) -- (221.27,273.17) ;
		\draw    (71.27,273.17) .. controls (73.7,286.07) and (196.7,285.07) .. (201.27,273.17) ;
		\draw    (201.27,253.17) -- (240.97,253.17) ;
		\draw    (221.27,253.17) -- (221.27,263.17) ;
		\draw    (221.27,243.17) -- (221.27,253.17) ;
		\draw    (61.27,253.17) -- (61.27,263.17) ;
		\draw    (51.27,253.17) -- (51.27,263.17) ;
		\draw    (51.27,243.17) -- (51.27,253.17) ;
		\draw    (411.27,253.17) -- (411.27,263.17) ;
		\draw  [draw opacity=0][fill={rgb, 255:red, 208; green, 2; blue, 27 }  ,fill opacity=0.48 ] (191.27,253.17) -- (221.27,253.17) -- (221.27,243.17) -- (191.27,243.17) -- cycle ;
		\draw  [draw opacity=0][fill={rgb, 255:red, 208; green, 2; blue, 27 }  ,fill opacity=0.48 ] (51.27,233.17) -- (411.27,233.17) -- (411.27,213.17) -- (51.27,213.17) -- cycle ;
		\draw    (411.27,223.17) -- (411.27,233.17) ;
		\draw    (311.27,223.17) -- (311.27,233.17) ;
		\draw    (151.27,223.17) -- (151.27,233.17) ;
		\draw    (51.27,223.17) -- (51.27,233.17) ;
		\draw    (51.27,213.17) -- (51.27,223.17) ;
		\draw    (411.27,213.17) -- (411.27,223.17) ;
		\draw    (50.97,223.17) -- (411.27,223.17) ;
		\draw    (50.97,213.17) -- (411.27,213.17) ;
		\draw  [draw opacity=0][fill={rgb, 255:red, 74; green, 144; blue, 226 }  ,fill opacity=0.8 ] (71.27,273.17) -- (91.27,273.17) -- (91.27,263.17) -- (71.27,263.17) -- cycle ;
		\draw  [draw opacity=0][fill={rgb, 255:red, 74; green, 144; blue, 226 }  ,fill opacity=0.8 ] (191.27,273.17) -- (201.27,273.17) -- (201.27,263.17) -- (191.27,263.17) -- cycle ;
		\draw  [draw opacity=0][fill={rgb, 255:red, 74; green, 144; blue, 226 }  ,fill opacity=0.8 ] (381.27,273.17) -- (391.27,273.17) -- (391.27,263.17) -- (381.27,263.17) -- cycle ;
		\draw  [draw opacity=0][fill={rgb, 255:red, 74; green, 144; blue, 226 }  ,fill opacity=0.8 ] (261.27,273.17) -- (281.27,273.17) -- (281.27,263.17) -- (261.27,263.17) -- cycle ;
		\draw  [draw opacity=0][fill={rgb, 255:red, 74; green, 144; blue, 226 }  ,fill opacity=0.8 ] (91.27,263.17) -- (101.27,263.17) -- (101.27,253.17) -- (91.27,253.17) -- cycle ;
		\draw  [draw opacity=0][fill={rgb, 255:red, 74; green, 144; blue, 226 }  ,fill opacity=0.8 ] (281.27,263.17) -- (291.27,263.17) -- (291.27,253.17) -- (281.27,253.17) -- cycle ;
		\draw  [draw opacity=0][fill={rgb, 255:red, 74; green, 144; blue, 226 }  ,fill opacity=0.8 ] (171.27,263.17) -- (191.27,263.17) -- (191.27,253.17) -- (171.27,253.17) -- cycle ;
		\draw  [draw opacity=0][fill={rgb, 255:red, 74; green, 144; blue, 226 }  ,fill opacity=0.8 ] (361.27,263.17) -- (381.27,263.17) -- (381.27,253.17) -- (361.27,253.17) -- cycle ;
		\draw  [draw opacity=0][fill={rgb, 255:red, 74; green, 144; blue, 226 }  ,fill opacity=0.8 ] (341.27,253.17) -- (361.27,253.17) -- (361.27,243.17) -- (341.27,243.17) -- cycle ;
		\draw  [draw opacity=0][fill={rgb, 255:red, 74; green, 144; blue, 226 }  ,fill opacity=0.8 ] (291.27,253.17) -- (311.27,253.17) -- (311.27,243.17) -- (291.27,243.17) -- cycle ;
		\draw  [draw opacity=0][fill={rgb, 255:red, 74; green, 144; blue, 226 }  ,fill opacity=0.8 ] (151.27,253.17) -- (171.27,253.17) -- (171.27,243.17) -- (151.27,243.17) -- cycle ;
		\draw  [draw opacity=0][fill={rgb, 255:red, 74; green, 144; blue, 226 }  ,fill opacity=0.8 ] (101.27,253.17) -- (121.27,253.17) -- (121.27,243.17) -- (101.27,243.17) -- cycle ;
		\draw  [draw opacity=0][fill={rgb, 255:red, 74; green, 144; blue, 226 }  ,fill opacity=0.8 ] (121.27,243.17) -- (151.27,243.17) -- (151.27,233.17) -- (121.27,233.17) -- cycle ;
		\draw  [draw opacity=0][fill={rgb, 255:red, 74; green, 144; blue, 226 }  ,fill opacity=0.8 ] (311.27,243.17) -- (341.27,243.17) -- (341.27,233.17) -- (311.27,233.17) -- cycle ;
		\draw    (261.27,273.17) .. controls (263.7,286.07) and (386.7,285.07) .. (391.27,273.17) ;
		
		\draw (112.27,284) node [anchor=north west][inner sep=0.75pt]    {$s[ p\dd q]$};
		\draw (301.27,283) node [anchor=north west][inner sep=0.75pt]    {$s[ p'\dd q']$};

	\end{tikzpicture}
	
	\end{center}
	\vskip-6mm
	\caption{Here $s[p\dd q] = s[p'\dd q']$. The blocks depicted as blue form the fingerprints of these strings; the red blocks intersect the strings but can be ``inconsistent'' and are not in the fingerprints.}
	\label{fig:fingerprints}
\end{figure}

We store in each block $B$ in the hierarchy built for $t$ pointers to children/parent of $B$ and to the first preceding block on the same level with different $\id(B)$. With this, one can retrieve the fingerprint for any substring of $t$ in time $O(\log m\log\log n)$.

Analogously, we define the deterministic fingerprint for any substring of $s$.

\begin{lemma}
	Two substrings of $s$ or $t$ coincide iff their deterministic fingerprints coincide.\label{lem:fingerprint}
\end{lemma}
\begin{proof}
	By construction, any two blocks in both hierarchies (for $s$ and for $t$) that have equal identifiers are equal as strings. Therefore, two substrings with equal fingerprints must be equal. For converse, suppose that two substrings $s[p'\dd q')$ and $t[p\dd q)$ are equal. It suffices to show that the procedure $\fin$ computing the fingerprints for $s[p'\dd q')$ and $t[p\dd q)$ does the same calculations. Suppose that, on a generic step, a call to $\fin(\ell, B_1,\ldots,B_b)$ occurs during the computation for $s[p'\dd q')$ and the same call $\fin(\ell, B_1,\ldots,B_b)$ for $t[p\dd q)$; initially, for $\ell = 0$, $\fin(0, s[p'],\ldots,s[q'{-}1])$ and $\fin(0, t[p],\ldots,t[q{-}1])$ are called on the same sequences of level-$0$ blocks. Cases 2 or 3 in the description of $\fin$ obviously produce the same sequences of identifiers surrounding a recursive call to $\fin$; the recursive call receives as its arguments a certain new sequence of blocks $B'_1,\ldots,B'_{b'}$. In case 2 this sequence $B'_1,\ldots,B'_{b'}$ is the same in both calls since the blocks $B'_1,\ldots,B'_{b'}$ either are copies of blocks from $B_1,\ldots,B_b$ or are runs of blocks to which we assigned the same identifiers of the form $\langle\id(B),r\rangle$ in both hierarchies of blocks. In case 3 the argument is more complicated but still straightforward since the blocks $B'_1,\ldots,B'_{b'}$ were produced from local minima of the function $\id''$ and we strategically skipped some blocks from the begin and end of the sequence $B_1,\ldots,B_b$ to make the local minima consistent, and, further, we rely on the fact that blocks created from the same sequences of blocks in both hierarchies are assigned with the same identifiers (recall that the trie $T_{\id}$ is used to guarantee this). The values of $\id''$ for $B_1,\ldots,B_b$ during the computation of $\fin$ are exactly the same as they were in the definition of both hierarchies of blocks (for $s$ and for $t$) except, possibly, for $B_1$ and $B_2$, as $\id''(B_1) = \id''(B_2) = \infty$ when $\id''$ was computed in $\fin$. But a simple case analysis shows that this nuance does not hurt the argument.
\end{proof}

By Lemma~\ref{lem:traversal}, the hierarchy of blocks $H$ for $s$ can be emulated using $J$, which suffices to get the fingerprint of any substring of $s$. But this approach is inefficient as it assembles the fingerprints top down, not bottom up as in the definition. We speed up the assembling by adding the following data structures.
Let $B$ be a block of $J$ (intermediate or not) with $|B| > 1$. The block $B$ must have been produced from a run/range/subrange of blocks $B_i,\ldots,B_j$ with $i < j$. We store in $B$ a link to a block in $J$ with identifier $\id(B_i)$, which exists in $J$ by Lemma~\ref{lem:leftmost}. The links naturally form a forest $A_{L}$ on all blocks. Define a symmetric forest $A_{R}$ for rightmost blocks (i.e., where the link from $B$ refers to $\id(B_j)$). We equip $A_{L}$ and $A_{R}$ with the weighted ancestor structure \cite{GLN,KopelowitzLewenstein} that, for any $B$ and any threshold $w$, can find the farthest ancestor $\hat{B}$ of $B$ in $A_L$ (or $A_R$) with $|\hat{B}| \ge w$ in $O(\log\log n)$ time. 
Using $A_L$ and $A_R$, we obtain the following result; its tedious proof is very technical and \ifdefined\fullpaper it was moved to  Appendix~\ref{appx:search}.\else it can be found in the full version of the paper  \cite[Appendix~C]{KosolobovFull}.\fi

\begin{restatable}[fingerprints]{lemma}{fingerprintsubstr}
	Given a substring $s[p\dd q]$ and the lowest block $B$ in the hierarchy of blocks $H$ for $s$ with $\sbeg(B) \le p < q \le \send(B)$, if $J$ retains $B$ and $B$ has children in $J$, then one can compute the fingerprint of $s[p\dd q]$ in $O(\log m \cdot(\log\log n)^2)$ time, where $m = q-p$.\label{lem:fingeprintsubstr}
\end{restatable}

\textbf{Indexing structures.}
We use a well-known scheme for pattern matching. The following description is sketchy in parts as it is mostly the same as in \cite{ChristiansenEtAl,ClaudeNavarro}. Given $s[p\dd p{+}|t|) = t$, let $B$ be the lowest block in $J$ such that $\sbeg(B) \le p < p + |t| - 1 \le \send(B)$. If $B$ has children in $J$, we call the occurrence \emph{primary}; otherwise, \emph{secondary}. We first find primary occurrences. 

Let $\lvec{T}$ and $T$ be two compacted tries. Consider each block $B$ in $J$. Let $B$ be a non-run block with children $B_i,\ldots,B_j$ in $J$. For each $h \in [i\dd j)$, we store the string $B_{h+1}$ in $T$ and $T$ has an explicit node $x$ such that $\str(x) = B_{h+1}$; we store $\lrange{B_i\cdots B_{h}}$ in $\lvec{T}$ with an explicit node $y$ such that $\str(y) = \lrange{B_i\cdots B_{h}}$. Thus, we produce pairs of nodes $(x,y)$ associated with $B$ and indices $h$. Let $B$ be a run block whose only child in $J$ is $B_1$ and $r = |B| / |B_1|$. We store the string $B_1$ in $T$ and $\lvec{B}_1{}^{r-1}$ in $\lvec{T}$, thus producing a pair of nodes $(x,y)$ associated with $B$. The edges in $T$ and $\lvec{T}$ store only lengths of edge labels.  By Theorem~\ref{thm:space}, the size of  $\lvec{T}$ and $T$ and the number of produced pairs $(x,y)$ is $O(\delta\log\frac{n}{\delta})$. Each node $z$ in $T$ stores a pointer to a block $B$ in $J$ such that $\str(z)$ is a prefix of $B$; each node $z$ in $\lvec{T}$ stores an index $h$ and a pointer to $B$ in $J$ whose children are $B_i,\ldots,B_j$ such that $\str(z)$ is a prefix of $\lrange{B_i\cdots B_h}$.

Define an order on nodes of $T$ and $\lvec{T}$: $x < x'$ iff $\str(x) < \str(x')$. Fix a constant $\epsilon > 0$. We store a range reporting structure $R$ \cite{ChanLarsenPatrascu}: it contains all produced pairs of nodes $(x,y)$ in $O(\delta\log\frac{n}{\delta})$ space (the big-O hides factor $\frac{1}{\epsilon}$) and, given any nodes $a,b$ of $T$ and $c,d$ of $\lvec{T}$, we can report all $N$ pairs $(x,y)$ such that $a \le x \le b$ and $c \le y \le d$ in $O((1 + N)\log^\epsilon n)$ time.

We turn $T$ into a z-fast trie by augmenting it with a hash table that, for each node $x$, maps the fingerprint of a certain prefix of $\str(x)$ to $x$. We use our deterministic fingerprints. Since each fingerprint takes $\omega(1)$ space, the hash table is organized as a compacted trie $T_f$ containing the fingerprints, treated as sequences of identifiers of length $O(\log m\log\log n)$; the edges of $T_f$ do not store their edge labels, only lengths (as in the trie $T_{\id}$). Let $t[q\dd m)$ be a string that can be read by descending from the root of $T$, assuming that the edge labels are somehow accessible by an ``oracle''. The z-fast trie can read it in $O(\log^2 m \log\log n)$ time by querying the hash on $O(\log m)$ fingerprints of prefixes of $t[q\dd m)$. If $t[q\dd m)$ cannot be read from the root of $T$, the z-fast trie returns an arbitrary node $x$ of $T$. The found node $x$ contains a pointer to a block $B$ such that $\str(x)$ is a prefix of $B$. Using Lemma~\ref{lem:fingeprintsubstr} (fingerprints) and relying on Lemma~\ref{lem:fingerprint}, we check in $O(\log m\cdot (\log\log n)^2)$ time whether the prefix of length $m - q$ of the string $B$ is equal to $t[q\dd m)$ and, thus, we verify whether we indeed performed the descending correctly. We equip $\lvec{T}$ with the same z-fast trie structure.

\textbf{\boldmath Searching for $t$.}
Let $s[p\dd p{+}|t|) = t$ be a primary occurrence. Let $B$ be the lowest block in $J$ such that $\sbeg(B) \le p < p + |t| - 1 \le \send(B)$ and let $B_i,\ldots,B_j$ be all children of $B$ in $J$. Choose the largest $h \in [i\dd j)$ such that $\send(B_h) < p + |t| - 1$. If $B$ is not a run block, there is $q \in [0\dd m)$ such that $\sbeg(B_{h+1}) = p + q$ and $T$ contains the string $t[q\dd m)$ and $\lvec{T}$ contains the string $\lrange{t[0\dd q)}$; the range reporting structure $R$ contains a pair of nodes $(x,y)$ associated with $B$ and index $h$. If $B$ is a run block, there is $q \in [0\dd m)$ such that $m - q \le |B_1|$, $\sbeg(B) + c|B_1| = p + q$, for some $c \ge 1$, and $T$ contains $t[q\dd m)$ and $\lvec{T}$ contains $\lrange{t[0\dd q)}$; there is again an associated pair of nodes $(x,y)$. Given the position $q$, using the fingerprints of prefixes of $\lrange{t[0\dd q)}$ and $t[q\dd m)$ in the z-fast tries, one can find in $O(\log^2 m\cdot(\log\log n)^2)$ time the nodes $a,b$ in $T$ and $c,d$ in $\lvec{T}$ on which the range reporting structure $R$ will report all pairs $(x,y)$ that fit the above description and each such reported pair will be associated with a separate primary occurrence (note that a run block $B$ will contain a group of primary occurrences at distances that are multiples of $|B_1|$ but the algorithm finds only the rightmost occurrence from each group; other group occurrences are easy to restore using periodicity).

To identify candidate positions $q \in [0\dd m)$ for which we perform the described search, we compute the fingerprint $f$ of $t$, which is a sequence of $O(\log m\log\log n)$ block identifiers. Let $p_1,\ldots,p_b$ be the starting positions of these blocks from $f$ such that $0 = p_1 < \cdots < p_b < p_{b+1} = m$, i.e., for $i \in [1\dd b]$, the $i$th block from $f$ corresponds to $t[p_i\dd p_{i+1})$. By Lemma~\ref{lem:fingerprint}, each primary occurrence $s[p\dd p{+}|t|)$ of $t$ has the same fingerprint and its respective blocks correspond to the substrings $s[p{+}p_i\dd p{+}p_{i+1})$, for $i \in [1\dd b]$. Let $B$ be the lowest block in $J$ such that $\sbeg(B) \le p < p + |t| - 1 \le \send(B)$ and let $B_i,\ldots,B_j$ be all children of $B$ in the hierarchy of blocks $H$ for $s$ (not in $J$!). Any substring $s[p + p_i \dd p + p_{i+1})$, for $i \in [1\dd b]$, that corresponds to a non-run block in the fingerprint must form a block in $H$, which must be a descendant of $B$ in $H$. Any substring $s[p + p_i \dd p + p_{i+1})$ that corresponds to a run block with an identifier $\langle\id(B'_i),r\rangle$, for some block $B'_i$ and $r > 1$, must form a sequence of blocks $s[p + p_i + c|B'_i|\dd p + p_i + (c + 1)|B'_i|)$ in $H$ with $c \in [0\dd (p_{i+1} - p_i) / |B'_i|)$, all of which must be descendants of $B$ in $H$. These observations imply that any block boundary between $B_i,\ldots,B_j$ that is inside the range $[p\dd p{+}|t|)$ is equal either to $p + p_i$, for some $i \in [1\dd b]$, or to $p + p_{i+1} - |B'_i|$, where $t[p_i\dd p_{i+1})$ corresponds to a run block with an identifier $\langle B'_i, r\rangle$. We obtain $O(\log m\log\log n)$ candidate positions for $q$: all these $p_i$ and $p_{i+1} - |B'_i|$.

\textbf{Secondary occurrences.}
To find secondary occurrences, we reverse all pointers in the tree $J$ and traverse the reversed pointers starting from each occurrence (including newly added secondary occurrences). 
More precisely, any non-intermediate leaf block $B$ in $J$ with $|B|> 1$ contains a pointer to another block $\hat{B}$ in $J$ such that $\id(\hat{B}) = \id(B)$; any intermediate block $B$ in $J$ contains two numbers $h = \mathsf{off}(B)\ge 0$ and $r = \mathsf{r}(B) > 1$ and a pointer to a block $\hat{B}$ from a level $2k + 2$  such that $\id(B)$ was created from identifiers $\id(B_{i+h}),\ldots,\id(B_{i+h+r-1})$, where  $B_{i},\ldots,B_{j}$ is the sequence of level-$(2k + 1)$ blocks from which $\hat{B}$ was created and $i + h + r - 1 \le j$. The block $\hat{B}$ has children in $J$ in both cases. Given a leaf $B$, secondary occurrences of $t$ corresponding to $B$ are exactly all occurrences of $t$ from the substring $s[\sbeg(B)\dd \send(B)]$. They may exist if the block $\hat{B}$ referred by $B$ contains occurrences of $t$. 

We construct on the tree $J$ the marked ancestor data structure \cite{AlstrupHusfeldtRauhe} that allows us to mark some nodes and to search the nearest marked ancestor of any node in $O(1)$ time. Then, we reverse all pointers: for each block $\hat{B}$, we store in $\hat{B}$ a link to each non-intermediate block $B$ that has a pointer to $\hat{B}$ and we mark $\hat{B}$ in this case; further, if $\hat{B}$ is from a level $2k + 2$ and $\hat{B}$ was created from level-$(2k + 1)$ blocks $B_{i},\ldots,B_{j}$, then we store in $\hat{B}$ a link to each intermediate block $B$ that has a pointer to $\hat{B}$ and we store the range $[i + h \dd i + h + r)$ corresponding to $B$ in an augmented binary search tree $R_{\hat{B}}$ associated with $\hat{B}$; for any $d \in [i\dd j]$, we mark the block $B_d$ if $B_d$ is a non-intermediate child of $\hat{B}$ in $J$ and $d \in [i + h \dd i + h + r)$ for a range $[i + h\dd i + h + r)$ from $R_{\hat{B}}$. Note that the tree $R_{\hat{B}}$ contains at most $O((\log\log n)^2)$ distinct ranges but one range may correspond to many blocks $B$; thus, range searching operations on $R_{\hat{B}}$ can be performed in time $O(\log\log\log n) \le O(\log^\epsilon n)$.

With this machinery, the search for secondary occurrences is just a traversal of the reversed pointers. We maintain a set of occurrences of $t$, which initially contains only primary occurrences, and each occurrence $s[p\dd p{+}|t|)$ is associated with a lowest block $\hat{B}$ in $J$ such that $\sbeg(\hat{B}) \le p < p + |t| - 1 \le \send(\hat{B})$. Given such an occurrence  $s[p\dd p{+}|t|)$ and the associated block $\hat{B}$, we use the tree $R_{\hat{B}}$ to check whether $s[p\dd p{+}|t|)$ lies inside some blocks $B_{i+h}\cdots B_{i+h+r-1}$ where $[i + h \dd i + h + r)$ is a range stored in $R_{\hat{B}}$ and $B_i,\ldots,B_j$ are blocks from which $\hat{B}$ was created in the hierarchy of blocks $H$ for $s$. Each link associated with such range induces a new unique secondary occurrence of $t$ that is added to the set of occurrences. We then loop through all links to all non-intermediate blocks $B$ that have pointers to $\hat{B}$ and each such link analogously induces a new secondary occurrence. Then, we find the nearest marked ancestor $B'$ of $B$ (if any) and continue the same process recursively with $B'$ in place of $B$, inducing new unique secondary occurrences. Each added secondary occurrence in the set is processed in the same way.
The time for reporting all $\mathrm{occ}$ occurrences is $O(\mathrm{occ}\cdot\log^\epsilon n)$.

\textbf{Time.}
We obtain $u = O(\log m \log\log n)$ candidate positions $q$ at which we perform z-fast trie searches and range reporting: $O(\log^2 m\cdot (\log\log n)^2 + (1 + \mathrm{occ}_q)\log^\epsilon n)$ time per $q$, where $\mathrm{occ}_q$ is the number of primary occurrences found for one $q$. We have $\mathrm{occ} \ge \sum_q \mathrm{occ}_q$. Thus, the total time is $O(m + u(\log^2 m\cdot (\log\log n)^2  + \log^\epsilon n) + \mathrm{occ}\cdot\log^\epsilon n) = O(m + \log^3 m\cdot (\log\log n)^3 + \log m \cdot\log\log n \cdot \log^\epsilon n + \mathrm{occ}\cdot\log^\epsilon n)$. Fix a constant $\epsilon' > \epsilon$. When $m \le \log^6 n$, we estimate $\log m = O(\log\log n)$ and, thus, the sum $\log^3 m\cdot (\log\log n)^3 + \log m \cdot\log\log n\cdot \log^\epsilon n$ is bounded by $O(\log^{\epsilon'} n)$ and, after renaming $\epsilon'$ to $\epsilon$, we obtain time $O(m + (1 + \mathrm{occ})\log^\epsilon n)$. When $m > \log^6 n$, we bound this sum by $O(m)$ and obtain time $O(m + \mathrm{occ}\cdot\log^\epsilon n)$.

\section{Construction}
\label{sec:construction}

We build the following components: the compacted tries {\boldmath$T_{\id}$}, {\boldmath$\lvec{T}$}, {\boldmath${T}$}, {\boldmath$T_f$}, the range reporting structure {\boldmath$R$} and all its pairs {\boldmath$(x,y)$}, the trees {\boldmath$A_L$} and {\boldmath$A_R$} with a weighted ancestor structure, and, for each block $B$, pointers to children (if any), the pointer to a copy of $B$ (if $B$ is a leaf), and {\boldmath $\sbeg(B)$, $\send(B)$, $\ell(B)$, $\id(B)$} (if any). In the end, we also construct in an obvious way all back links for secondary occurrences with a marked ancestor structure \cite{AlstrupHusfeldtRauhe}.

Our algorithm reads $s$ from left to right and maintains for each level of the currently constructed hierarchy $H$ a queue with at most $O(\log\log n)$ last built blocks from this level. The letters $s[0], s[1], \ldots$ are consecutively fed to the level $0$. Each level $\ell$ receives new blocks from the previous level $\ell - 1$ or from $s$, in case of level $0$, and puts them to the right of its queue (the queue elements are ordered from left to right and we dequeue elements from left); when a block is received, we may remove some blocks from the queue and may generate new blocks to feed for level $\ell + 1$. If the level $\ell + 1$ does not exist, we create it. After reading the whole $s$, we feed into level $0$ a special letter $\$$ that forces the algorithm to complete all levels.

At any level $\ell$, we assume that all level-$\ell$ blocks of the hierarchy $H$ to the left of the leftmost block in the queue were already built together with their descendants and structures like $A_L$, $A_R$, $T_{\id}$. In particular, we can use Lemma~\ref{lem:fingeprintsubstr} (fingerprints) on these blocks. We maintain a global hash table $L$ that maps any $\id$ to the leftmost block in $J$ with this $\id$ (as in Lemma~\ref{lem:leftmost}). By Lemma~\ref{lem:leftmost}, such leftmost blocks have children in $J$ and are never removed. Thus, once a block with new $\id$ is created and the corresponding information about it and its children is added to $L$, $T_{\id}$, $\lvec{T}$, $T$, $T_f$, $A_L$, $A_R$, this information will never be removed.

\textbf{Time.}
We use hash tables for $L$, edge transitions in $T_{\id}$, $T$, $\lvec{T}$, $T_f$, for van Emde Boas structures from the weighted ancestor structure \cite{GLN,KopelowitzLewenstein} of $A_L$ and $A_R$, and for other places. The query time in the tables is $O(1)$ and the insertion time is amortized $O(1)$ w.h.p.~\cite{BFKK} It is the only component that uses randomization. From now on, we do not write explicitly that the time for modifying the tables is ``amortized expected w.h.p.'', assuming it by default.

For level $\ell = 2k$ or $\ell = 2k+1$, we will spend $O(\log^3 |B|\cdot \log n)$ time on queue operations for each level-$\ell$ block $B$ with $|B| \le 2^k$, $O(\log^3 |B|\cdot \log n)$ time on each new block $B$ that we will generate for level $\ell+1$, and $O(1)$ time on each level-$\ell$ block $B$ with $|B| > 2^k$ by simply feeding such $B$ to level $\ell + 1$. By Lemma~\ref{lem:sparsity}, estimating $\log^3 |B| = O(k^3)$ when $|B| \le 2^k$, the total time for all levels can be bounded by $O(\sum_{k=0}^\infty \frac{n}{2^k} k^3 \log n) = O(n\log n \cdot\sum_{k=0}^\infty \frac{k^3}{2^k}) = O(n\log n)$.

\textbf{\boldmath Level $\ell = 2k$, for $k \ge 0$.} We maintain at most one block $B_1$ in the queue and a number $r$ such that the last $r$ fed blocks had identifiers $\id(B_1)$. Suppose that a new block $B_2$ is fed to us. If $\id(B_2) = \id(B_1)$, increment $r$, done. Let $\id(B_2) \ne \id(B_1)$.  We remove $B_1$ from the queue and process it as described below (if the queue is not empty), after which we finalize as follows: put $B_2$ in the queue seting $r := 1$, if $|B_2| \le 2^k$, or feed $B_2$ to level $\ell + 1$, if $|B_2| > 2^k$; done. Now let us describe how $B_1$ and $r$ are processed before the finalization.  

If $r = 1$, feed $B_1$ to level $\ell + 1$, finalize. If $r > 1$, create a new block $B$ with $\id(B) = \langle\id(B_1),r\rangle$, $\sbeg(B) = \sbeg(B_1)$, $\send(B) = \sbeg(B) + r|B_1| - 1$, $\ell(B) = \ell$ (which is correct by Lemma~\ref{lem:uniformity}) and insert $B$ in $A_L$ and $A_R$ in $O(\log\log n)$ time \cite{KopelowitzLewenstein} with a link to the leftmost block with identifier $\id(B_1)$ found using $L$.  If $L$ finds a block $B'$ with $\id(B') = \id(B)$, store in $B$ a pointer to $B'$, feed $B$ to level $\ell + 1$, finalize. If $L$ finds no such $B'$, we proceed as follows. 

We insert $B$ into $L$ and set $B_1$ as the only child of $B$. We can search any prefix of $B_1$ in the z-fast trie $T$ in $O(\log^2 |B_1|\cdot(\log\log n)^2)$ time as in the search phase of Section~\ref{sec:search}, using Lemma~\ref{lem:fingeprintsubstr} (fingerprints) for prefixes of $B_1$; thus, by the binary search on prefixes of $B_1$, we find the location in $T$ at which the string $B_1$ should be inserted in $O(\log^3 |B_1| (\log\log n)^2)$ time. We insert $B_1$ by creating a node $x$ in $T$ such that $\str(x) = B_1$ (if it did not exist) and, then, we modify the z-fast trie structure by inserting $O(1)$ new fingerprints into $T_f$: to insert a fingerprint in $T_f$, we descend from the root reading it in $T_f$, skipping edge labels, and then we restore the skipped labels from a reference to an appropriated substring whose fingerprint was stored in $T_f$, which takes $O(\log |B_1|(\log\log n)^2)$ time. Similarly, we insert the string $\lvec{B}_1{}^{r-1}$ into $\lvec{T}$, creating a node $y$, in $O(\log^3 |B|\cdot (\log\log n)^2)$ time. Finally, we feed $B$ to level $\ell + 1$ and finalize. We add the pair $(x,y)$ (together with a link to the block $B$) to the set of pairs for the range reporting structure $R$, which is built in the very end of the construction in $O(N\log N)$ time, for all $N$ pairs \cite{ChanLarsenPatrascu}. Since $\delta\log\frac{n}{\delta} = O(n)$, this time is $O(n\log n)$.

\textbf{\boldmath Level $\ell = 2k + 1$, for $k{\ge}0$.}
Let $B_1,B_2,\ldots$ be all level-$\ell$ blocks from left to right, marked according to conditions (a)--(b) of Section~\ref{sec:blocks}: $B_h$ is marked iff $|B_h| > 2^k$ or $|B_{h+1}| > 2^k$ or $\id''(B_{h-1})$ is a local minimum. Suppose that a new block $B_j$ is fed to us. We compute $\id''(B_j)$ using the numbers $\id(B_{j-1})$, $\id'(B_{j-1})$. The queue contains the last contiguous sequence of unmarked blocks $B_i,\ldots,B_{j-1}$ fed to level $\ell$ (i.e., $B_{i-1}$ is marked or $i = 1$). We determine whether $B_{j-1}$ is marked after $B_j$ is received. As in Section~\ref{sec:blocks}, we have $j - i = O(\log\log n)$.

If $|B_j| \le 2^k$ and both $B_{j-1}$ and $B_j$ are not marked, put $B_j$ in the queue, done. If $|B_j| \le 2^k$ and $B_j$ is marked ($B_{j-1}$ cannot be marked in this case), we process $B_i,\ldots,B_j$ and unite them into a new block that is then fed to level $\ell + 1$; we then take $B_i,\ldots,B_j$ off the queue, done.
If $|B_j| > 2^k$, we mark $B_{j-1}$ and analogously process and unite $B_i,\ldots,B_{j-1}$ instead of $B_i,\ldots,B_j$; then, we move $B_j$ to level $\ell + 1$, done (the queue is empty in the end). Let us describe how $B_i,\ldots,B_j$ are processed and united (for $B_i,\ldots,B_{j-1}$, it is analogous).

We read the sequence $\id(B_i),\ldots,\id(B_j)$ in the trie $T_{\id}$: by descending and skipping edge labels, we reach a node $z$ and restore the skipped labels using a block $B'$ referred by $z$ as described in Lemma~\ref{lem:traversal}. Then, if we have successfully verified that $\str(z)$ equals the sequence, we remove $B_i,\ldots,B_j$ and create a new childless block $B$ with $\ell(B) = \ell$, $\sbeg(B) = \sbeg(B_i)$, $\send(B) = \send(B_j)$ and we store in $B$ a pointer to $B'$. We insert $B$ in $A_L$ and $A_R$ with links to leftmost blocks with identifiers $\id(B_i)$ and $\id(B_j)$. We then move $B$ to level $\ell + 1$, done.

Suppose that we could not find the above block $B'$ in $T_{\id}$. We create a new block $B$ with a new identifier $\id(B)$ and we insert the sequence $\id(B_i),\ldots,\id(B_j)$ into $T_{\id}$  in $O(\log\log n)$ time, creating a node $z$ that refers to $B$ (the algorithm is almost the same as reading this sequence in $T_{\id}$). We insert $B$ in $A_L$ and $A_R$ with links to leftmost blocks with identifiers $\id(B_i)$ and $\id(B_j)$. Then, we must greedily parse $B_i,\ldots,B_{j}$ into subranges: if $B_i,\ldots,B_{m-1}$ have already been parsed (for $m\ge i$), then $B_{m},\ldots,B_{m+r-1}$ is the longest subrange with $m+r-1 \le j$ for which there is $m' < m$ such that $\rleft(m',4r) = \rleft(m,4r)$ and $\rright(m',5r) = \rright(m,5r)$; let $r = 1$ if there is no such $r > 0$. (We use the notation $\rleft$ and $\rright$ here as in Section~\ref{sec:blocks}.) For each subrange $B_{m},\ldots,B_{m+r-1}$ with $r > 1$, we form a new intermediate block $\tilde{B}$ in its place, removing the blocks $B_{m},\ldots,B_{m+r-1}$, and we identify the smallest $m''$ such that $\rright(m'',r) = \rright(m,r)$; we store in $\tilde{B}$ certain numbers $\mathsf{off}(\tilde{B})$, $\mathsf{r}(\tilde{B}) = r$, and a pointer to the parent of $B_{m''}$ required to find the blocks $B_{m''},\ldots,B_{m''+r-1}$ (see details in Section~\ref{sec:jiggly}).
After creating all intermediate blocks, the block $B$ has children $\tilde{B}_{\tilde{\imath}},\ldots,\tilde{B}_{\tilde{\jmath}}$ (intermediate or not) and we do almost the same computations as on level $2k$: for each $h \in [\tilde{\imath}\dd \tilde{\jmath})$, we store the string $\lrange{\tilde{B}_{\tilde{\imath}}\dots \tilde{B}_h}$ in $\lvec{T}$, producing a node $x$, and $\tilde{B}_{h+1}$ in ${T}$, producing a node $y$, which also requires a modification of $T_f$; we add the pair $(x,y)$ to $R$. All this, except the computation of subranges, takes $O(\log^3 |B|\cdot (\log\log n)^3)$ overall time. We then move $B$ to level $\ell + 1$, done. It remains to describe how $B_i,\ldots,B_j$ are parsed into the subranges, which is subtle.

\textbf{\boldmath Parsing $B_i,\ldots,B_j$.} 
We maintain two compacted tries $T_\circ$ and $\lvec{T}_\circ$ defined as follows. Let $B'$ be a block of $J$ from a level $2k' + 2$ and $B'_{i'},\ldots,B'_{j'}$ be the level-$(2k' + 1)$ blocks from which $B' = B'_{i'}\cdots B'_{j'}$ was created. As $B_i,\ldots,B_j$, the sequence $B'_{i'},\ldots,B'_{j'}$ was also greedily parsed into subranges: for each such subrange $B'_{m'},\ldots,B_{m'+r'-1}$ with $m' \in [i'\dd j']$, we store the sequence $\id(B'_{m'}),\id(B'_{m'+1}),\ldots,\id(B'_{m'+r'-1})$ in the trie $T_\circ$ and $T_\circ$ contains an explicit node $x$ such that $\str(x)$ equals this sequence, and we store $\id(B'_{m'-1}),\id(B'_{m'-2}),\ldots,\id(B'_{i'}), \$$ in $\lvec{T}_\circ$ with an explicit node $y$ in $\lvec{T}_\circ$ with $\str(y)$ equal to this sequence. All such pairs $(x,y)$ are stored in a set $P$. The nodes of $T_\circ$ and $\lvec{T}_\circ$ store pointers to (leftmost) blocks of $J$ from which their edge labels can be restored. 
Each node $x \in \lvec{T}_\circ$ (in $T_\circ$, analogously) stores an $O(\log n)$-bit number $\mathsf{v}(x)$ such that, for any $x,y\in \lvec{T}_\circ$, we have $\mathsf{v}(x) < \mathsf{v}(y)$ iff $\str(x) < \str(y)$ (lexicographically); the numbers can be maintained with $O(1)$ amortized insertion time \cite{BCDFZ,DietzSleator} (such structure is called dynamic ordered linked list). Define $x < y$ iff $\mathsf{v}(x) < \mathsf{v}(y)$. We store in each node $x$ the smallest and largest nodes, $\mathsf{L}(x)$ and $\mathsf{R}(x)$, in the subtree rooted at $x$.

We equip $P$ with a dynamic range reporting structure \cite{Blelloch,Nekrich}: it uses $O(\log n)$ time for insertions and $O(\log n + t\frac{\log n}{\log\log n})$ for queries, where $t$ is the number of reported points. This structure on $P$ needs to know only the relative order of coordinates of its points $(x,y)$, not the values $\mathsf{v}(x)$, $\mathsf{v}(y)$. By Theorem~\ref{thm:space}, the size of $T_\circ$, $\lvec{T}_\circ$, $P$ is $O(\delta\log\frac{n}{\delta})$ (the total number of pairs $(x,y)$). Once we have parsed $B_i,\ldots,B_j$ into subranges, the updates for $T_\circ$, $\lvec{T}_\circ$, $P$ are straightforward (similar to updates for $T_{\id}$ and $R$ above) and take $O((j-i+1)\log n)$ time.

Suppose that we have already parsed $B_{i},\ldots,B_{m-1}$ into subranges. Let us find the largest $r$ for which there is $m' < m$ with $\rleft(m',4r) = \rleft(m,4r)$ and $\rright(m',5r) = \rright(m,5r)$. Initially, set $r = 1$. We consecutively process the positions $h = i,\ldots,j$ as follows. We read the sequence $\id(B_h),\ldots,\id(B_{\min\{j,m+5(r+1)-1\}})$ by descending from the root of $T_\circ$, thus reaching a node $z$; set $z = \nil$ if we fail. We read the sequence $\id(B_{h-1}), \ldots, \id(B_{\max\{i-1,m-4(r+1)\}})$ from $\lvec{T}_\circ$, thus reaching a node $v$, where we abuse notation by assuming that $\id(B_{i-1})$ equals the special symbol $\$$; set $v = \nil$ if we fail. 
If either $h \le m - 4(r + 1)$ and $z \ne \nil$, or $h \ge m + 5(r + 1)$ and $v \ne \nil$, then we increment $r$ by $1$ and repeat the processing of $h$ again. Otherwise, we continue as follows. If $z = \nil$ or $v = \nil$, we proceed to the next $h$; otherwise, we perform the range query in $P$ on $[\mathsf{L}(z)\dd \mathsf{R}(z)]\times [\mathsf{L}(v)\dd \mathsf{R}(v)]$: if it reports some points, we increment $r$ by $1$ and repeat the processing of $h$ again, otherwise, we proceed to the next $h$.

The reading in $T_\circ$/$\lvec{T}_\circ$ takes $O(\log\log n)$ time. The queries in $P$ are performed only when $h \in [m-4(r+1)\dd m+5(r+1))$. Hence, $r$ was computed in time $O((j - i)\log\log n + r\log n) = O((\log\log n)^2 + r \log n) = O(r \log n)$. 
The algorithm is correct by the same reasons why our substring search was correct. 
Finally, we try to extend $r$ by a naive pattern matching of the sequence $\rleft(m,4(r+1)), \rright(m,5(r+1))$ in $\id(B_i),\ldots,\id(B_j)$ in $O(r\log\log n)$ time. Thus, $r$ is found and it remains to compute the smallest $m''$ such that $\rright(m'',r) = \rright(m,r)$.

For each node $z \in T_\circ$, denote $P_z = P \cap [\mathsf{L}(z)\dd \mathsf{R}(z)]\times [-\infty\dd \infty]$ and $\mathsf{N}(z) =|P_z|$. See Fig.~\ref{fig:chunks} in Appendix \ifdefined\fullpaper\ref{appx:search}.\else\ref{appx:jiggly}.\fi\ For each $z$ with $\mathsf{N}(z) > \log\log n$, we order $P_z$ according to the $y$-coordinates of its points and we split $P_z$ into  disjoint chunks $\{P_z^{a}\}_a$, each chunk $P_z^a$ of size $\Theta(\log\log n)$, so that all $y$-coordinates in $P_z^{a}$ are smaller than in $P_z^{a'}$ whenever $a < a'$. We store in $z$ a van Emde Boas structure $V_z$: for each $P_z^{a}$, $V_z$ stores the minimum and maximum of $\{\mathsf{v}(y) \colon (x,y)\in P_z^a\}$. 
We assign to each point $(x,y) \in P$ the position $p_{x,y}$ such that, at the time of insertion of $(x,y)$ into $P$, the node $x$ corresponded to a prefix of $s[p_{x,y}\dd n)$ and $y$ to a suffix of $s[0\dd p_{x,y})$. We store in $z$ a dynamic linear-space RMQ structure $Q_z$ from \cite{BrodalDavoodiRao2} that contains a linked list of all chunks $P_z^{a}$, where each $P_z^a$ stores the number $\min\{p_{x,y} \colon (x,y) \in P_z^a\}$, and $Q_z$ supports insertions/updates in $O(\frac{\log n}{\log\log n})$ time and, for any $b$ and $c$, $Q_z$ can compute $\min\{p_{x,y}\colon$ $(x,y) \in \bigcup_{a\in[b\dd c]} P_z^a\}$ in $O(\frac{\log n}{\log\log n})$ time. Note that we spent $O(1)$ space per chunk $P_z^a$ in $z$.

To compute $m''$, we process each $h \in (m\dd m{+}r)$ as follows: read the sequence $\rright(h,m+r-h)$ in $T_\circ$ and $\rleft(h,h-m)$ in $\lvec{T}_\circ$, thus reaching nodes $z$ and $v$ (as in the above analysis). If $\mathsf{N}(z) \le \log\log n$, we report all $t$ points in the range $[\mathsf{L}(z)\dd \mathsf{R}(z)]\times [\mathsf{L}(v)\dd \mathsf{R}(v)]$ in $P$ in $O(\log n + t\frac{\log n}{\log\log n})$ time, which is $O(\log n)$ as $t \le \mathsf{N}(z)$; each point corresponds to an occurrence of the sequence $\id(B_{m}),\ldots,\id(B_{m+r-1})$ and we choose the point $(x,y)$ with minimal $p_{x,y}$. If $\mathsf{N}(z) > \log\log n$, we use $V_z$ to find all chunks $P_z^a$ such that $P_z^a \subseteq [\mathsf{L}(z)\dd \mathsf{R}(z)]\times [\mathsf{L}(v)\dd \mathsf{R}(v)]$ and we use $Q_z$ to compute $\min p_{x,y}$ for all $(x,y)$ from such chunks; two chunks $P_z^a$ may partially intersect this range and, for each such $P_z^a$, we calculate all elements $P_z^a\cap [\mathsf{L}(z)\dd \mathsf{R}(z)]\times [\mathsf{L}(v)\dd \mathsf{R}(v)]$ using one range reporting query on $P$. This overall takes time $O(\log n + |P_z^a|\cdot\frac{\log n}{\log\log n}) = O(\log n)$.

The total space is as follows. For a given node $z$, all its data structures take $O(|P_z| / \log\log n)$ space. Fix a depth $\ell$ in the trie $T_\circ$. The sets $P_z$ for all nodes $z$ with depth $\ell$ partition (a subset of) $P$ and, therefore, in total occupy $O(|P| / \log\log n)$ space. Since the maximal depth of $T_\circ$ is $O(\log\log n)$, the total space is $O(|P|)$.

We maintain all $V_z$ and $Q_z$ as follows. Suppose that we insert a point $(x,y)$ in $P$. We increment $\mathsf{N}(z)$ for all parents $z$ of $x$. If $\mathsf{N}(z) = \log\log n$, we create $V_z$ and $Q_z$ for $z$ with only one chunk. If $\mathsf{N}(z) > \log\log n$, we insert $(x,y)$ into one chunk $P_z^a$ found using $V_z$, possibly updating $V_z$  in $O(\log\log n)$ time and $Q_z$ in $O(\frac{\log n}{\log\log n})$ (if $p_{x,y} < \min\{p_{\tilde{x},\tilde{y}}  \colon (\tilde{x},\tilde{y}) \in P_z^a\}$); if $P_z^a$ overflows, we split it into two chunks with further updates to $Q_z$. To split $P_z^a$, we obtain all elements of $P_z^a$ by one range reporting query on $P$ in $O(\log n + |P_z^a|\frac{\log n}{\log\log n}) = O(\log n)$ time. 
Thus, the amortized time of insertion in $P_z$ is $O(\frac{\log n}{\log\log n})$ since each chunk $P_z^a$, after its creation, can take $\Theta(\log\log n)$ more insertions before it is split. Hence, the amortized time of insertions in all parents of $x$ in $T_\circ$ is $O(\frac{\log n}{\log\log n}\cdot\mathsf{height}(T_\circ)) = O(\log n)$.

When the value $\mathsf{v}(v)$ of $v\in\lvec{T}_\circ$ changes, we should update all structures $V_z$ containing it. To do it efficiently, we store in $v$ as many copies of $v$ as there are points of the form $(z,v)$ in $P$, each with its own $\mathsf{v}(v)$ and they all are in the ordered linked list with all nodes of $T_\circ$. This does not change the algorithm or space bounds. Let $v$ correspond to a point $(z,v) \in P$. In each parent $z'$ of $z$ in $T_\circ$, we update $V_{z'}$, which takes $O(\log\log n)$ time (to find whether the old value $\mathsf{v}(v)$ is present in $V_{z'}$ and to update it); $O((\log\log n)^2)$ time overall. Since the numbers $\mathsf{v}(\cdot)$ change at most $O(n)$ times in total, the time for all updates is $O(n(\log\log n)^2)$.

\textbf{Removed blocks.} Blocks $B_i,\ldots,B_j$ on level $2k + 1$ or $B_1$ on level $2k$ can be removed. Non-leftmost blocks in $J$ have no children and store pointers to their leftmost copies, which, by Lemma~\ref{lem:leftmost}, cannot be removed. Hence, no ``cleanup'' is needed after removing blocks; $A_L$ and $A_R$ need no changes since their blocks contain links only to undeletable leftmost blocks.

\textbf{Correctness.}
The described algorithm exactly corresponds to the construction of the hierarchy $H$ from Section~\ref{sec:blocks} and it applies all rules 1--4 whenever possible and parses ranges of children blocks into subranges whenever possible. Hence, it indeed produces $J$.

\begin{theorem}
	Suppose that $s$ is a string of length $n$ over an alphabet $\{0,1,\ldots,n^{O(1)}\}$ and $\delta$ is its string complexity. Using $O(\delta \log\frac{n}{\delta})$ machine words, one can construct in one left-to-right pass on $s$ in a streaming fashion an index of size $O(\delta \log\frac{n}{\delta})$ that can find in $s$ all $\mathrm{occ}$ occurrences of any string of length $m$ in $O(m + (\mathrm{occ} + 1)\log^\epsilon n)$ time, where $\epsilon > 0$ is any fixed constant (the big-$O$ in space hides factor $\frac{1}{\epsilon}$). The construction time is $O(n\log n)$ w.h.p.
\end{theorem}



\bibliography{refs}


\appendix
\newpage

\section{Hierarchy of Blocks (Missing Proofs)}
\label{appx:hierarchy}

\consistency*
\begin{proof}
	The proof is by induction on $\ell$. For $\ell = 0$, it is trivial. Assume $\ell > 0$. Let $\ell' = \ell - 1$. Denote $\tilde{\imath} = \sbeg({B}_h)$ and $\tilde{\jmath} = \send({B}_h)$. Suppose that ${B}_h = \hat{B}_u \cdots \hat{B}_v$, where $\hat{B}_u,\ldots,\hat{B}_v$ are level-$\ell'$ blocks corresponding to ${B}_h$. 
	By the inductive hypothesis, there are level-$\ell'$ blocks $\hat{B}_{u'},\ldots,\hat{B}_{v'}$ such that $v'-u' = v-u$ and these blocks correspond to the substring $s[i\dd j]$ in the sense that $i = \sbeg(\hat{B}_{u'})$ and $\id(\hat{B}_{u'}) = \id(\hat{B}_u),\ldots,\id(\hat{B}_{v'}) = \id(\hat{B}_v)$. Note that, for the inductive hypothesis to hold, it was sufficient to have a weaker assumption that $s[i{-}2^{k+3}\dd j{+}2^{k+3}] = s[\tilde{\imath}{-}2^{k+3}\dd \tilde{\jmath}{+}2^{k+3}]$.
	
	Since the construction process copies blocks of length greater than $2^k$ from level $\ell'$ to level $\ell$ unaltered, all blocks $\hat{B}_u,\ldots,\hat{B}_v$  either have length at most $2^k$, or $u = v$ and $|\hat{B}_u| > 2^k$. In the latter case, we have $u' = v'$, $\id(\hat{B}_{u'}) = \id(\hat{B}_u) = \id({B}_h)$ and the block $\hat{B}_{u'}$ will be copied to level $\ell$ in the same way as $B_u$ had been copied to level $\ell$, which implies that the required level-$\ell$ block ${B}_{h'}$ is the copy of $\hat{B}_{u'}$.
	
	From now on, assume that each of the blocks $\hat{B}_u,\ldots,\hat{B}_v$ has length at most $2^k$. Observe that, since $0 < i \le j < |s|-1$ and $0 < \tilde{\imath} \le \tilde{\jmath} < |s|-1$, neither of the blocks $\hat{B}_u, \hat{B}_v, \hat{B}_{u'}, \hat{B}_{v'}$ is the first or last block on level $\ell'$. 
	
	Let $\ell = 2k+1$. We will prove the inductive step under the following weaker assumption: $s[i{-}2^{k+3}{-}2^k\dd j{+}2^{k+3}{+}2^k] = s[\tilde{\imath}{-}2^{k+3}{-}2^k\dd \tilde{\jmath}{+}2^{k+3}{+}2^k]$ (it will be useful in the analysis of the case $\ell = 2k + 2$). As it was noted above, such weaker assumption still implies the existence of the blocks $\hat{B}_{u'},\ldots,\hat{B}_{v'}$ corresponding to the substring $s[i\dd j]$. We have $\id(\hat{B}_u) = \cdots = \id(\hat{B}_v)$ and $\id(\hat{B}_{u-1}) \ne \id(\hat{B}_u)$ and  $\id(\hat{B}_v) \ne \id(\hat{B}_{v+1})$. Since $\id(\hat{B}_{u'}) = \id(\hat{B}_u), \ldots, \id(\hat{B}_{v'}) = \id(\hat{B}_v)$, it remains to prove that $\id(\hat{B}_{u'-1}) \ne \id(\hat{B}_{u'})$ and $\id(\hat{B}_{v'}) \ne \id(\hat{B}_{v'+1})$. If $|\hat{B}_{u-1}| \le 2^k$, then the substring $\hat{B}_{u-1} = s[\tilde{\imath}{-}|\hat{B}_{u-1}|\dd \tilde{\imath})$ is inside $s[\tilde{\imath}{-}2^{k+3}{-}2^{k}\dd \tilde{\jmath}{+}2^{k+3}{+}2^{k}]$ together with a ``neighborhood'' of length $2^{k+3}$ since $|\hat{B}_{u-1}| + 2^{k+3} \le 2^k + 2^{k+3}$. Therefore, by the inductive hypothesis, the substring $s[i{-}|\hat{B}_{u-1}|\dd i)$, which is equal to $\hat{B}_{u-1}$, must form a level-$\ell'$ block $\hat{B}_{u'-1}$ with $\id(\hat{B}_{u'-1}) = \id(\hat{B}_{u-1})$. Hence, we obtain $\id(B_{u'-1}) \ne \id(\hat{B}_{u'})$ since $\id(\hat{B}_{u-1}) \ne \id(\hat{B}_u)$. By the same reasoning, if $|\hat{B}_{u'-1}| > 2^k$, then $|\hat{B}_{u-1}| > 2^k$ and vice versa (due to symmetry). In case $|\hat{B}_{u'-1}| > 2^k$, we have $\id(\hat{B}_{u'-1}) \ne \id(\hat{B}_{u'})$ since $|\hat{B}_{u'}| \le 2^k$. Symmetrically, one can show that $\id(\hat{B}_{v'}) \ne \id(\hat{B}_{v'+1})$. 
	
	Let $\ell = 2k + 2$. To generate level-$\ell$ blocks (including ${B}_h$), the process marked some blocks from level $2k+1$ using rules (a) and (b) described in Section~\ref{sec:blocks}. According to these rules, $\hat{B}_{u-1}$ and $\hat{B}_v$ must be marked and $\hat{B}_u,\ldots,\hat{B}_{v-1}$ must all be unmarked. It suffices to prove that $\hat{B}_{u'-1}$ and $\hat{B}_{v'}$ are marked and $\hat{B}_{u'},\ldots,\hat{B}_{v'-1}$ are unmarked. Let us first derive a certain technical observation.
	
	Suppose that, for some $t \in [1\dd 7]$, each of the blocks $\hat{B}_{u-t},\ldots,\hat{B}_{u-1}$ has length at most $2^k$. Since $t\cdot2^k + 2^{k+3} + 2^k \le 2^{k+4}$, each of these blocks is inside the substring $s[\tilde{\imath}{-}2^{k+4}\dd \tilde{\jmath}{+}2^{k+4}]$ together with a ``neighborhood'' of length $2^{k+3} + 2^k$. As was proved in case $\ell = 2k + 1$, it follows that the substrings $s[i{-}|\hat{B}_{u-t}\cdots \hat{B}_{u-1}|\dd i{-}|\hat{B}_{u-t+1}\cdots \hat{B}_{u-1}|),\ldots,s[i{-}|\hat{B}_{u-1}|\dd i)$, which are equal to $\hat{B}_{u-t},\ldots,\hat{B}_{u-1}$, respectively, must form level-$\ell'$ blocks $\hat{B}_{u'-t},\ldots,\hat{B}_{u'-1}$ such that $\id(\hat{B}_{u'-t}) = \id(\hat{B}_{u-t}), \ldots, \id(\hat{B}_{u'-1}) = \id(\hat{B}_{u-1})$. Symmetrically, if each of the blocks $\hat{B}_{u'-t},\ldots,\hat{B}_{u'-1}$ has length at most $2^k$, then also $\id(\hat{B}_{u'-t}) = \id(\hat{B}_{u-t}), \ldots, \id(\hat{B}_{u'-1}) = \id(\hat{B}_{u-1})$.
	
	Suppose that $\hat{B}_{u-1}$ was marked because of rule (b), i.e., $\infty > \id''(\hat{B}_{u-3}) > \id''(\hat{B}_{u-2})$ and $\id''(\hat{B}_{u-2}) < \id''(\hat{B}_{u-1}) < \infty$. Hence, each of the five blocks $\hat{B}_{u-5},\ldots,\hat{B}_{u-1}$ has length at most $2^k$. Due to the above observation, we have $\id(\hat{B}_{u'-5}) = \id(\hat{B}_{u-5}),\ldots,\id(\hat{B}_{u'-1}) = \id(\hat{B}_{u-1})$. Therefore, the sequence $\id''(\hat{B}_{u-3}),\ldots,\id''(\hat{B}_v)$ is equal to the sequence $\id''(\hat{B}_{u'-3}),\ldots,\id''(\hat{B}_{v'})$. Hence, $\hat{B}_{u'-1}$ must be marked and $\hat{B}_{u'},\ldots,\hat{B}_{v'-1}$ must be unmarked. We can analogously prove that, if $\hat{B}_v$ was marked due to rule (b), then $\hat{B}_{v'}$ is also marked.
	
	Suppose that $\hat{B}_{u-1}$ was marked because of rule (a), i.e., $|\hat{B}_{u-1}| > 2^k$ or $|\hat{B}_{u}| > 2^k$. The latter is impossible since we assumed that $|\hat{B}_u| \le 2^k$. Now, again due to the above observation, we have $|\hat{B}_{u'-1}| > 2^k$. Thus, $\id''(\hat{B}_{u'-1}) = \infty$ and, hence, $\hat{B}_{u'-1}$ is marked. Further, the sequences $\id''(\hat{B}_{u-1}),\ldots,\id''(\hat{B}_v)$ and $\id''(\hat{B}_{u'-1}),\ldots,\id''(\hat{B}_{v'})$ coincide, which implies that $\hat{B}_{u'},\ldots,\hat{B}_{v'-1}$ are unmarked and, again, if $\hat{B}_v$ was marked due to rule (b), then $\hat{B}_{v'}$ is marked.
	
	It remains to consider the case when $\hat{B}_v$ is marked due to rule (b) and to prove that $\hat{B}_{v'}$ must be marked too in this case. We assumed $|\hat{B}_v| \le 2^k$ and, hence, the only option is that $|\hat{B}_{v+1}| > 2^k$. But we can repeat the proof of the above observation to deduce that $|\hat{B}_{v+1}| \le 2^k$ iff $|\hat{B}_{v'+1}| \le 2^k$. Therefore, we obtain $|\hat{B}_{v'+1}| > 2^k$ and, thus, $\hat{B}_{v'}$ is marked. 
\end{proof}

\ifdefined\fullpaper
\sparsity*
\begin{proof}
	Denote by $S_k$ the set of block boundaries for level $2k$ (and $S_{k+1}$ is for level $2k+2$). 	It suffices to prove the claim only for $S_k$ as the set of block boundaries for level $2k+1$ is a subset of $S_k$. 
	As described in Section~\ref{sec:blocks}, during the transition from level $2k + 1$ to $2k + 2$, every maximal range $B_x,\ldots,B_y$ of blocks of level $2k+1$ that consists only of blocks of length ${\le}2^k$ was split into subranges according to certain rules discussed in Section~\ref{sec:blocks}, and each such subrange formed a block for level $2k+2$. The key observation is that each of the subranges of $B_x,\ldots,B_y$ contains at least two blocks, except the last subrange that might be equal to $B_y$. Therefore, any two consecutive block boundaries $j$ and $j'$ from level $2k+1$ (i.e., $j < j'$ and  there are no block boundaries in $(j\dd j')$) can belong to $S_{k+1}$ only if either $j' - j > 2^k$ or $j'' - j' > 2^k$ where $j'' > j'$ is the block boundary following $j'$ (assuming $j'' = \infty$ if $j'$ is the last position in $s$).
	
	Let us prove by induction on $k$ the following claim: the range $[0\dd n)$ of positions of the string $s$ can be partitioned into disjoint non-empty chunks $F_1, F_2, \ldots, F_m$ such that $F_i = [f_i\dd f_{i+1})$, where  $f_{m+1} = n$ and $f_1 < f_2 < \cdots < f_{m+1}$, and the chunks are of two types:
	\begin{enumerate}
		\item a chunk $F_i$ is \emph{normal} if $2^{k-5} \le |F_i| \le 4\cdot 2^{k-5}$ and $|F_i \cap S_k| \le 2$;
		\item a chunk $F_i$ is \emph{skewed} if $|F_i| \ge 8\cdot 2^{k-5}$, $|F_i \cap S_k| \le 3$, and $[f_i{+}4\cdot 2^{k-5}\dd f_{i+1}) \cap S_k = \emptyset$, i.e., all positions of $F_i \cap S_k$ are concentrated in the prefix of $F_i$ of length $4\cdot 2^{k-5}$ (hence the name ``skewed'').
	\end{enumerate}
	A chunk $F_i$ is called empty if $F_i\cap S_k = \emptyset$.
	The claim implies the lemma as follows. The worst case, which maximizes $|S_k \cap [i\dd j)|$, occurs when a range $[i\dd j)$ intersects one skewed chunk and at most $\lceil\frac{j-i-3}{2^{k-5}}\rceil$ normal chunks in such a way that the skewed chunk is the rightmost one and contains three positions from $S_k$, all inside $[i\dd j)$, and each normal chunk contains two positions from $S_k$, all inside $[i\dd j)$. We then have $|S_k \cap [i\dd j)| \le 2\lceil\frac{j-i-3}{2^{k-5}}\rceil + 3$. The bound can be rewritten as $2\lceil \frac{c 2^k + d}{2^{k-5}}\rceil + 3$, substituting integers $c \ge 0$ and $d \in [0\dd 2^k)$ such that $j - i - 3 = c 2^k + d$, and it is upper-bounded by $2^6c + \lceil \frac{d}{2^{k-5}}\rceil + 3 \le 2^6 c + 2^5 + 3 < 2^6(c + 1) = 2^6\lceil\frac{c2^k + 1}{2^k}\rceil \le 2^6 \lceil\frac{(j - i - 3) + 1}{2^k}\rceil \le 2^6  \lceil\frac{j - i}{2^k}\rceil$, as the lemma states.
	
	Now let us return to the inductive proof of the claim.
	The base of the induction is $k \le 6$ and it is trivial since the range $[0\dd n)$ always can be split into chunks of length $2$. For the inductive step, suppose that the inductive hypothesis holds for an integer $k \ge 6$ and let us construct a partitioning of the range $[0\dd n)$ into chunks for $k + 1$.
	
	First, we greedily unite chunks into disjoint pairs and singletons from right to left as follows: given a chunk $F_i$ such that the chunks $F_{i+1}, F_{i+2}, \ldots$ were already united into pairs and singletons, we consider the following cases: (1)~if $F_i$ is a skewed chunk, then it forms a singleton and we analyze $F_{i-1}$ next; (2)~if both $F_{i-1}$ and $F_{i}$ are normal chunks, then $F_{i-1}$ and $F_{i}$ are united into a pair and we consider $F_{i-2}$ next; (3)~if $F_i$ is a normal chunk and $F_{i-1}$ is a skewed chunk, then we cut from $F_{i-1}$ a new normal chunk $F' = [f_{i}{-}2^{k-5}\dd f_{i})$ of length $2^{k-5}$, which is necessarily empty due to properties of skewed chunks, and we pair $F'$ with $F_i$ and analyze the (already cut) chunk $F_{i-1}$ next; the length of the skewed chunk $F_{i-1}$ is reduced by $2^{k-5}$ after the cutting and might become less than $8\cdot 2^{k-5}$ but it is still considered as skewed since this is not an issue as the chunk $F_{i-1}$ will be anyways dissolved into normal chunks in a moment. After the construction of all the pairs and singletons, we proceed as follows.\footnote{There is a subtle special case when $F_1$ is normal and it cannot be unitied into a pair since it is leftmost. We then add a ``dummy'' empty chunk $F_0$ that has length $2^{k-5}$ and covers the ``negative'' positions $-1,-2,\ldots,-2^{k-5}$. The dummy chunks and such coverage of negative positions, which never contain block boundaries, will not pose problems for the overall analysis.}
	
	We consider all the produced pairs and singletons of chunks from right to left. A singleton chunk $F_i$ is always a skewed chunk and its length is at least $7\cdot 2^{k-5}$ (a suffix of length $2^{k-5}$ could be cut from the chunk). Let $S'_k = \{j_1 < \cdots < j_{|S'_k|}\}$ be the set of block boundaries from level $2k + 1$. It follows from the key observation above that any two consecutive positions $j_{h-1}$ and $j_{h}$ from $S'_{k}$ can both belong to the set $S_{k+1}$ only if either $j_{h} - j_{h-1} > 2^{k}$, or $j_{h+1} - j_{h} > 2^{k}$. In both cases there is a gap of length greater that $2^{k}$ between either $j_{h-1}$ and $j_{h}$, or $j_{h}$ and $j_{h+1}$. Therefore, since only the first $4\cdot 2^{k-5} = 2^{k-3}$ positions of the skewed chunk $F_i$ may contain positions of $S'_k$, the transition to level $2k + 2$ in the case $|F_i \cap S'_k| = 3$ necessarily removes from $S'_k$ either the first or the second position of $S'_k \cap F_i$ when producing $S_{k+1}$. Thus, we have $|F_i \cap S_{k+1}| \le 2$. We then split the skewed chunk $F_i$ into a series of new normal chunks all of which, except possibly the first one, do not contain positions from $S_{k+1}$: the first chunk is $[f_{i}\dd f_{i}{+}4\cdot 2^{k-5})$ and it has length $4\cdot 2^{k-5} = 2\cdot 2^{k-4}$; the remaining suffix of $F_i$ is $[f_{i}{+}4\cdot 2^{k-5}\dd f_{i+1})$ (the boundary $f_{i+1}$ used here can differ from the original $f_{i+1}$ if the skewed chunk $F_i$ was cut and, thus, $f_{i+1}$ was decreased) and its length is at least $7\cdot 2^{k-5} - 4\cdot 2^{k-5}= 3\cdot 2^{k-5}$, this suffix is split into normal chunks arbitrarily (recall that the length of a normal chunk for the inductive step $k+1$ must be at least $2^{k-4}$ and at most $4\cdot 2^{k-4}$).
	
	Consider a pair of normal chunks $F_{i-1}$ and $F_{i}$. For simplicity, we denote by $F_{i-1}$ the chunk preceding $F_i$ even if it was created by cutting the skewed chunk (actual $F_{i-1}$) that preceded $F_i$ in the initial partitioning. The length of the united chunk $F_{i-1} F_{i}$ is at least $2^{k-4}$ and at most $4\cdot 2^{k-4}$ so that it could have formed a new normal chunk if at most two positions of $S_{k+1}$ belonged to it. By the key observation above, if two consecutive positions $j_{h-1}$ and $j_{h}$ were retained in $S_{k+1}$ after the transition to level $2k+2$, then there is a gap of length greater than $2^{k}$ between either $j_{h-1}$ and $j_{h}$, or $j_{h}$ and $j_{h+1}$. Therefore, since $|F_{i-1}F_{i}| \le 4\cdot 2^{k-4} < 2^{k}$, the united chunk $F_{i-1} F_{i}$ contains at most three positions of $S_{k+1}$, i.e., one of the positions from $S'_k \cap F_{i-1}F_{i}$ must have been discarded. In case $|S_{k+1} \cap F_{i-1}F_{i}| \le 2$, we simply form a new normal chunk $F_{i-1}F_{i}$. The case $|S_{k+1} \cap F_{i-1}F_{i}| = 3$ is more interesting; it occurs when each of the chunks $F_{i-1}$ and $F_{i}$ contains two positions from $S'_k$ and, during the transition to level $2k+2$, the second position of $S'_k \cap F_{i-1}$ was removed and the two positions from $S'_k \cap F_{i}$ were retained. By the key observation above, we must have a gap after the chunk $F_i$ in this case: $S'_{k} \cap [f_{i+1}\dd f_{i+1}{+}2^{k}{-}|F_i|) = \emptyset$. Since $|F_i| \le 4\cdot 2^{k-5} = 2^{k-3}$, we hence obtain $S'_{k} \cap [f_{i+1}\dd f_{i+1}{+}7\cdot 2^{k-3}) = \emptyset$. Therefore, all newly produced chunks that are entirely contained in the range $[f_{i+1}\dd f_{i+1}{+}7\cdot 2^{k-3})$ are empty, i.e., they do not contain positions of $S'_k$ (recall that we process pairs and singletons of chunks from right to left and, so, only newly constructed chunks are located to the right of $F_i$). Let $\hat{F}_{i+1}, \hat{F}_{i+2}, \ldots, \hat{F}_{\ell}$ be a maximal sequence of consecutive newly constructed empty chunks to the right of $F_i$ (we use the notation $\hat{F}_j$ for the new chunks to distinguish them from the ``old'' chunks $F_1, F_2, \ldots$). We unite all chunks $\hat{F}_{i+1}, \hat{F}_{i+2}, \ldots, \hat{F}_{\ell}$  with $F_i$ and $F_{i-1}$ thus producing a new skewed chunk whose length is at least $|F_{i-1} F_{i}| + 7\cdot 2^{k-3} - 4\cdot 2^{k-4}$ (the negative term is because some chunks to the right of $F_i$ may only partially intersect the ``empty'' range $[f_{i+1}\dd f_{i+1}{+}7\cdot 2^{k-3})$), which is at least $2^{k-4} + 5\cdot 2^{k-3} = 11\cdot 2^{k-4} \ge 8\cdot 2^{k-4}$.
\end{proof}
\fi

\uniformity*
\begin{proof}
	The equality as strings is true by construction.
	Suppose, to the contrary, that the blocks $B$ and $B'$ were created on levels $\ell$ and $\ell'$, respectively, with $\ell < \ell'$ and $\ell$ is the smallest such level where the lemma fails. Let $\ell = 2k + 2$, for some $k$, and $B$ was created from blocks $B_i,\ldots, B_j$ from level $2k + 1$ so that $B = B_i\cdots B_j$. 
	Since $\id(B) = \id(B')$, there are blocks $B'_{i'},\ldots, B'_{j'}$ on a level ${<}\ell'$ from which the block $B'$ was created such that $j - i = j' - i'$ and $\id(B_i) = \id(B'_{i'}),\ldots,\id(B_{j}) = \id(B'_{j'})$. Since $2k+1 < \ell$ and $\ell$ is the minimal level where the lemma fails, $B'_{i'}\cdots B'_{j'}$ were present on level $2k+1$. The length of each of the blocks $B_i,\ldots,B_j$ is at most $2^k$ and, then, the same holds for $B'_{i'},\ldots,B'_{j'}$. Thus, the latter blocks should have participated in a maximal range of blocks with lengths ${\le}2^k$. In such a range, all blocks are united into disjoint subranges, each of which contains at least two blocks, except, possibly, the last subrange, which can be formed by one block. Each subrange generates a block for level $2k+2$. But this contradicts the assumption that the blocks $B'_{i'},\ldots,B'_{j'}$ should have been transitioned unaffected to level $\ell'-1 \ge \ell = 2k + 2$. 
	
	Let $\ell = 2k + 1$ for some $k$. One can analogously find decompositions of $B$ and $B'$ into blocks from level $2k$: $B = B_i\cdots B_j$ and $B' = B'_{i'}\cdots B'_{j'}$, where $j - i = j'- i'$ and  $\id(B_i) = \id(B'_{i'}),\ldots,\id(B_{j}) = \id(B'_{j'})$. Since $\ell = 2k + 1$, we have $\id(B_i) = \cdots =\id(B_j)$ and $B$ was produced by uniting these blocks into a ``run''. Then, $B'_{i'},\ldots,B'_{j'}$ should have been similarly united on the level $\ell - 1$, a contradiction.
\end{proof}

\section{Jiggly Block Tree (Missing Proofs)}
\label{appx:jiggly}

\leftmost*
\begin{proof}
	Consider the process that constructs $J$ by pruning the original hierarchy tree $H$ applying, first, rules 1--4 and, then, introducing intermediate blocks. Suppose, to the contrary, that a certain rule either removed $B$ itself or all children of $B$ and assume that $B$ is the first leftmost and topmost block with any $\id$ for which this happened. Let $\ell$ denote the level of $B$.
	
	The rules 1--2 obviously could not do this (recall that $|B| > 1$). Rule 3 cannot be applied to $B$ itself since $B$ has no ``preceding'' blocks with equal $\id$. Rule 4 applied to $B$ itself does not remove all children and, thus, is not contradictory. 
		
	One can show by induction on levels that, for any blocks $B'$ and $B''$, if $\id(B') = \id(B'')$, then, given $\ell' = 2k$ or $\ell' = 2k + 1$, any level-$\ell'$ block from the subtree rooted at $B'$ in the tree $H$ has a corresponding block with equal identifier in the subtree rooted at $B''$. It is then straightforward that a rule 3 or 4 applied to a block cannot remove a block $B'$ containing $B$ as a descendant since, for such a rule, there is always an unaffected block $B''$ ``preceding'' $B'$ with $\id(B'') = \id(B')$ and, hence, $B$ should have been among the descendants of $B''$ if $B$ were a descendant of $B'$, which is impossible because $B$ is already leftmost.
	
	Suppose that $B$ was pruned due to the introduction of an intermediate block $B'$ in place of a subrange of blocks $B_m\cdots B_{m+r-1}$ with $r > 1$ from a level $2k + 1$. By construction, there is $m' < m$ such that $\id(B_{m'+h}) = \id(B_{m+h})$, for all $h \in [0\dd r)$. If $B$ were a descendant of $B_{m+h}$, for $h \in [0\dd r)$ (maybe $B_{m+h} = B$), then $B$ must be a descendant of $B_{m'+h}$ in the tree $H$ by the observation discussed above. This is a contradiction since $B$ is leftmost.
\end{proof}

\spacetheorem*
\begin{proof}[Proof (continuation).]
	Let us detail the ideas described in the conceptual part of the proof from the main text.
	
	\textbf{Part i.} Fix $\ell = 2k+1$ or $\ell = 2k+2$, for some $k$. Consider a level-$\ell$ block $B$ in $J$ that is not a leaf. By Lemma~\ref{lem:sparsity} (local sparsity), there are $O(1)$ such blocks $B$ for which $\sbeg(B) < 2^{k+6}$ or $\send(B)  + 2^{k+6} \ge n$. Assume that these conditions do not hold for $B$ and, hence, the substring $s[\sbeg(B){-}2^{k+6}\dd \send(B){+}2^{k+6}]$ is well defined. 
	We have $|B| \le 2^k$ since otherwise $B$ should have had a copy on level $\ell+1$ that is the parent of $B$ in the original tree $H$ and, hence, $B$ could be removed by rule 2. Let us show that the string  $s[\sbeg(B){-}2^{k+4}\dd \send(B){+}2^{k+4}]$ has no occurrences at positions smaller than $\sbeg(B) - 2^{k+4}$. Suppose, to the contrary, that it occurs at a position $i' - 2^{k+4}$ with $i' < \sbeg(B)$. Then, due to Lemma~\ref{lem:consistency} (local consistency), there is a block $B' = s[i'\dd i'{+}|B|)$ in the tree $H$ such that $\id(B') = \id(B)$. By Lemma~\ref{lem:leftmost}, the tree $J$ contains the leftmost and topmost such block $B''$ with $\id(B'') = \id(B)$ and $\sbeg(B'') < \sbeg(B)$. But then all descendants of $B$ could be removed by rule 3, which is a contradiction. 
	
	The length of $s[\sbeg(B){-}2^{k+4}\dd \send(B){+}2^{k+4}]$ is $2\cdot 2^{k+4} + |B| \le 2^{k+5} + 2^k$. Hence, the strings $s[h\dd h{+}2^{k+6})$ with $h \in (\sbeg(B){-}2^{k+4}{-}2^k \dd \sbeg(B){-}2^{k+4}]$ all cover the substring $s[\sbeg(B){-}2^{k+4}\dd \send(B){+}2^{k+4}]$  and, thus, must be distinct and each such string $s[h\dd h{+}2^{k+6})$ has no occurrences at positions smaller than $h$. We assign the block $B$ to each of these $2^k$ strings $s[h\dd h{+}2^{k+6})$. We do the same analysis and assignments for every internal level-$\ell$ block in $J$. Due to  Lemma~\ref{lem:sparsity} (local sparsity), we could have assign at most $O(1)$ blocks to any fixed substring $s[h\dd h{+}2^{k+6})$ with $h \in [0\dd n{-}2^{k+6})$. Denote by $d_{2^{k+6}}$ the number of distinct substrings of length $2^{k+6}$ and by $b$ the number of internal blocks on level $\ell$ in $J$. We associate each distinct substring of length $2^{k+6}$ with its leftmost occurrence $s[h\dd h{+}2^{k+6})$. Thus, every internal block on level $\ell$ is assigned to exactly $2^k$ distinct substrings. We obtain $b\cdot 2^k \le O(d_{2^{k+6}})$. Hence, $b \le O(d_{2^{k+6}} / 2^k) \le O(\delta)$. So, the number of internal level-$\ell$ blocks in $J$ is $O(\delta)$ and the total number of internal blocks in $J$ is $O(\delta\log\frac{n}{\delta})$. Let us count leaves.
	
	\textbf{Part ii.}
	First, let us restrict our analysis to only certain leaf blocks $B$. Denote by $\mathsf{par}(B)$ the parent of $B$. Since the total number of internal blocks is $O(\delta\log\frac{n}{\delta})$, the number of leaf blocks $B$ such that either $2|B| \ge |\mathsf{par}(B)|$ or $B$ is the first or last child of $\mathsf{par}(B)$ is at most $O(\delta\log\frac{n}{\delta})$ (across all levels). Therefore, it suffices to consider the case when $2|B| < |\mathsf{par}(B)|$ and $B$ is neither first nor last child of $\mathsf{par}(B)$. Further, there are at most $O(\log n)$ blocks $B'$ such that $\sbeg(B') \le 2^{11}|B'|$ and $\send(B') + 2^{11}|B'| \ge n$ and each such block $B'$ contributes at most $O(\log\log n)$ leaf-children, i.e., $O(\log n \log\log n)$ in total, which is $O(\delta\log\frac{n}{\delta})$ since it was assumed that $\delta \ge \Omega(\log\log n)$. Thus, we can assume that $B' = \mathsf{par}(B)$, the parent of the leaf block $B$, is ``far enough'' from both ends of the string $s$ so that the substring $s[\sbeg(B'){-}2^{11}|B'| \dd\send(B'){+}2^{11}|B'|]$ is well defined. The block $B' = \mathsf{par}(B)$ is not a run block, due to rule 4; hence, $B$ is from a level $2k + 1$, for some $k$ ($B$ might be intermediate). 
	
	\textbf{($\star$) Restriction:} it suffices to count only leaf blocks $B$ from levels $2k + 1$ with $k \ge 0$ such that, for $B' = \mathsf{par}(B)$, we have $2|B| < |B'|$ and $B$ is neither first nor last child of $B'$, and $\sbeg(B') > 2^{11}|B'|$ and $\send(B') + 2^{11}|B'| < n$.
	
	Let $B$ be a leaf block from level $2k + 1$ as in the restriction and $B' = \mathsf{par}(B)$. Denote by $B_1,\ldots,B_b$ all blocks from level $2k + 1$ in the hierarchy $H$; we mark them according to conditions (a) and (b) from Section~\ref{sec:blocks} so that we can use the functions $\rleft$ and $\rright$ from now on. Let $B_i,\ldots,B_j$ be a range of blocks from which $B' = B_i\cdots B_j$ was produced in $H$. Thus, $B = B_m\cdots B_{m+r-1}$ for some $m$ and $r$ such that $[m\dd m{+}r) \subseteq (i\dd j)$. We have $m + r \le j$ since $B$ is not the last child of $B'$. The block $B$ is either intermediate (if $r > 1$) or an original block from $H$ (if $r = 1$). Denote $x = \sbeg(B)$ and $y = \send(B) + |B_{m+r}|$ and consider the substring $s[x\dd y] = B_m\cdots B_{m+r}$. The subrange $B_m,\ldots,B_{m+r-1}$ was distinguished by the greedy parsing procedure described in Section~\ref{sec:jiggly}, according to which there are no indices $m' < m$ such that $\rleft(m',4(r+1)) = \rleft(m,4(r+1))$ and $\rright(m', 5(r+1)) = \rright(m, 5(r+1))$. Note that the length of each of the blocks $B_i,\ldots,B_j$ is at most $2^k$. As $s[x\dd y] = B_m\cdots B_{m+r}$, one can show that the string $s[x{-}z\dd y{+}z]$ with $z = 4(r+2)2^{k} + 2^{k+4}$ cannot occur at any position $x'{-}z$ with $x' < x$ since, otherwise, Lemma~\ref{lem:consistency} (local consistency) would imply that there is $m' < m$ such that $\rleft(m',4(r+1)) = \rleft(m,4(r+1))$ and $\rright(m', 5(r+1)) = \rright(m, 5(r+1))$. Any block $B_h$ with $h\in [1\dd b]$ is marked iff $|B_h| > 2^k$ or $|B_{h+1}| > 2^k$ or $B_h$ is preceded by a local minimum $\id''(B_{h-1})$; to calculate the local minimum, we need the blocks $B_{h-4}$,$B_{h-3}$,$B_{h-2}$,$B_{h-1}$. If $\rleft(m, 4(r + 1))$ does not contain the special symbol $\$$, then $z$ is large enough to guarantee that at any occurrence $s[x'{-}z\dd y'{+}z]$ of $s[x{-}z\dd y{+}z]$ there are exactly the same $4(r+1)$ blocks $B_{m-4(r+1)}, \ldots, B_{m-1}$, each of length ${\le}2^k$, to the left of $s[x'\dd y']$ and, moreover, neither of them could be marked due to a local minimum since, if so, then the four blocks $B_{m-4(r+1)-4},\ldots,B_{m-4(r+1)-1}$, because of which the local minimum appeared, should have occured at both positions by Lemma~\ref{lem:consistency} (as $z$ is large enough). The case when $\rleft$ contains $\$$ is analogous and the analysis of $\rright$ is also analogous.
	
	Denote by $k'$ the smallest integer such that $k' \ge k + 4$ and $2^{k'} > |B_m\cdots B_{m+r}|$. It follows from Lemma~\ref{lem:sparsity} (local sparsity) that $r+2 \le 2^6 2^{k'} / 2^{k}$ and, hence, $2^{k'+6} \ge (r+2) 2^k$ and $2^{k'+8} + 2^{k'} \ge 4(r+2)2^k + 2^{k+4} = z$. Therefore, $2^{k'+10} > 2\cdot 2^{k'+8} + 2\cdot 2^{k'} \ge 2z + |B_m\cdots B_{m+r}| + 2^{k'}$ and the strings $s[h\dd h{+}2^{k'+10})$ with $h\in (x{-}z{-}2^{k'}\dd x{-}z]$ all cover $s[x{-}z\dd y{+}z]$ and, thus, must be distinct and each such string $s[h\dd h{+}2^{k'+10})$ has no occurrences at positions smaller than $h$ (note that all these strings are well defined since the block $B'$ is ``far enough'' from both ends of the string $s$). We assign the block $B$ to each of these $2^{k'}$ strings $s[h\dd h{+}2^{k'+10})$ with $h\in (x{-}z{-}2^{k'}\dd x{-}z]$. We do the same analysis and assignments for every leaf block $B$ in $J$ that satisfies the conditions from restriction ($\star$). 
	
	Fix a substring $s[h\dd h{+}2^{k'+10})$ with $h \in [0\dd n{-}2^{k'+10})$ and $k' > 0$. Let us show that $O(1)$ blocks were assigned to this substring in total. Let $k$ be such that $k' = k + 4$. Then, by Lemma~\ref{lem:sparsity} (local sparsity), at most $O(1)$ blocks were assigned from levels ${\ge}2k$. Let $k$ be such that $k' > k + 4$. If a block $B = B_m \cdots B_{m+r-1}$ from level $2k + 1$ was assigned to the string $s[h\dd h{+}2^{k'+10})$, where $B_m,\ldots,B_{m+r}$ are blocks from level $2k+1$ corresponding to $B$ as in the description above, then we have $2^{k'} > |B_m\cdots B_{m+r}| \ge 2^{k'-1}$. We associate the block $B$ with the interval $I(B) = [\sbeg(B_m)\dd \send(B_{m+r})]$ of length between $2^{k'-1}$ and $2^{k'}$ that is a subset of $[h\dd h{+}2^{k'+10})$. For each block $B$ assigned to $s[h\dd h{+}2^{k'+10})$ from every level $2k + 1$ such that $k' > k + 4$, we associate an interval $I(B)$ in this was. If intervals $I(B)$ and $I(B')$ intersect, then either they are nested (and the blocks $B$ and $B'$ are nested) or $B$ and $B'$ are from the same level. A given interval $I(B)$ may intersect at most two other intervals associated with blocks from the same level. If a block $B'$ was assigned to the same string as block $B$ and $B'$ is an ancestor of $B$ in $J$ (they are nested), then for any ancestor $B''$ of $B'$, $|B''| > 2|B'| \ge 2|B_m\cdots B_{m+r}| \ge 2^{k'}$ and, hence, $B''$ could not be assigned to the same string $s[h\dd h{+}2^{k'+10})$. So, the system of intervals $I(B)$ associated with all blocks $B$ assigned to the substring $s[h\dd h{+}2^{k'+10})$ contains $O(2^{k'+10} / 2^{k'}) = O(1)$ intervals: to see this, count separately the intervals that are not nested in any intervals, and all remaining intervals.
	
	Fix $k'$. Denote by $d_{2^{k'+10}}$ the number of distinct substrings of length $2^{k'+10}$ and by $\mathcal{B}_{k'}$ the set of blocks that were assigned to all strings $s[h\dd h{+}2^{k'+10})$ with $h \in [0\dd n{-}2^{k'+10})$. We associate each distinct substring of length $2^{k'+10}$ with its leftmost occurrence $s[h\dd h{+}2^{k'+10})$. Thus, every block from $\mathcal{B}_{k'}$ was assigned to exactly $2^{k'}$ such distinct substrings. We obtain $|\mathcal{B}_{k'}|\cdot 2^{k'} \le O(d_{2^{k'+10}})$ and, hence, $|\mathcal{B}_{k'}| \le O(d_{2^{k'+10}} / 2^{k'}) \le O(\delta)$. Therefore, one can estimate as $O(\delta\log\frac{n}{\delta})$ the number of leaf blocks that were assigned to substrings of length $2^{k'+10}$ with $k' \le \log\frac{n}{\delta}$ as described above, i.e., $\sum_{k'=1}^{\log(n/\delta)} |\mathcal{B}_{k'}| \le O(\delta\log\frac{n}{\delta})$. 
	
	\textbf{Part iii.} 
	It remains to count the leaf blocks that were assigned to substrings of length $2^{k'+10}$ with $k' > \log\frac{n}{\delta}$. Fix such $k'$. Consider a leaf block $B$ from a level $2k + 1$ that was assigned to a substring of length $2^{k'+10}$. Let $B = B_m \cdots B_{m+r-1}$, where $B_m,\ldots,B_{m+r}$ are blocks from level $2k+1$ corresponding to $B$ as in the description above.  We associate $B$ with the interval $I(B) = [\sbeg(B_m)\dd \send(B_{m+r})]$. We do the same association to intervals for all blocks $B$ assigned to substrings of length $2^{k'+10}$. As in the above discussion, one can show that all such intervals $I(B)$ are distinct, there is no ``double nestedness'' such as $I(B) \subset I(B') \subset I(B'')$, and an interval $I(B)$ can intersect at most two other intervals $I(B')$ and $I(B'')$ that are not nested in it and do not contain it and $B'$ and $B''$ must be from the same level as $B$ in this case. Since, for each such $B$, $|I(B)| \ge 2^{k'-1}$, one can estimate the number of intervals as $O(n/2^{k'})$ by counting separately the intervals that are not nested inside other intervals, and all remaining intervals. Therefore, the total number of blocks assigned to substrings of length $2^{k'+10}$ for all $k' > \log\frac{n}{\delta}$ is $O(\sum_{k'} n / 2^{k'}) = O(n / 2^{\log(n/\delta)}) = O(\delta)$.
\end{proof}

\ifdefined\fullpaper

\section{Substring Search (Missing Proofs)}
\label{appx:search}

\fingerprintsubstr*
\begin{proof}
	Suppose that the hierarchy of blocks $H$ is accessible. The fingerprint for $s[p\dd q]$ is a sequence of blocks $B^\circ_f, B^\circ_{f+1},\ldots,B^\circ_g$ from different levels of the hierarchy $H$ with $g - f = O(\log m\cdot\log\log n)$ and, for each $h \in [f\dd g)$, we have $\send(B^\circ_h) + 1 = \sbeg(B^\circ_{h+1})$. 
	
	Let us perform two depth first (DFS) traversals of the tree $H$, both starting from $B$, which we call ``left'' and ``right'': the left traversal descends from $B$ directly to the leaf block $s[p]$ and continues its work by recursively visiting children of each visited block in the left-to-right order, ignoring all blocks that end before position $p$, until the leaf $s[q]$ in the hierarchy $H$ is reached (at which point, the left traversal stops); the right traversal symmetrically descends to $s[q]$ and visits children right-to-left until the leaf $s[p]$ is reached. 
	
	Both traversals will visit all blocks $B^\circ_f, B^\circ_{f+1},\ldots,B^\circ_g$, provided the notion of ``visiting'' is defined properly. Namely, there is the following nuance here: any run block $B^\circ_h$ in the fingerprint whose identifier is of the form $\langle\id(B^\circ),r\rangle$, for some block $B^\circ$ and $r \ge 1$, actually might not be a block from $H$ but it has a corresponding run block $B'$ in $H$ such that $\id(B') = \langle\id(B^\circ),r'\rangle$, for some $r' \ge r$ (or it might be that $\id(B') = \id(B^\circ)$ if $r = 1$), and $\sbeg(B') \le \sbeg(B^\circ_h) \le \send(B^\circ_h) \le \send(B')$. We say that a traversal ``visits'' or ``encounters'' such block $B^\circ_h$ when it visits the corresponding block $B'$ in $H$. Our idea is to perform both traversals simultaneously and synchronously trying to prevent visits to any children of blocks $B^\circ_h$ with $h \in [f\dd g]$ (or blocks corresponding to $B^\circ_h$) that are encountered during the traversal, i.e., we try to cut the full traversal trees so that they will have exactly $g - f + 1$ leaves. (Note that we cautiously wrote above that we ``try'' to prevent visits to any children of $B^\circ_h$, for $h \in [f\dd g]$, because we indeed visit their children sometimes.)
	
	\textbf{The simultaneous traversal algorithm.}
	By definition, for some $\ell \ge 0$, we have $B^\circ_f,B^\circ_{f+1},\ldots,B^\circ_g = \ldots,\fin(\ell,B_1,\ldots,B_b),\ldots$ where $B_1,\ldots,B_b$ are some level-$\ell$ blocks that appeared during the computation of the fingerprint for $s[p\dd q]$ using the procedure $\fin$. Suppose that, on a generic step of the simultaneous traversal algorithm, we have visited all blocks that started before $\sbeg(B_1)$ in the left traversal and we have visited all blocks that ended after $\send(B_b)$ in the right traversal. Moreover, suppose that we have found and reported the parts of the fingerprint $B^\circ_f,B^\circ_{f+1},\ldots,B^\circ_g$ that were produced from blocks of $H$ from levels less than $\ell$, i.e., we have reported the blocks designated by the three dots ``$\ldots$'' before and after the sequence $\fin(\ell,B_1,\ldots,B_b)$ above. We continue the left traversal until the first block from level $\ell$ is visited and this block must be $B_1$; analogously, the right traversal finds $B_b$. The further algorithm depends on the value of $\ell$. 
	
	Case $\ell = 2k$, for some $k \ge 0$. If $|B_1| \le 2^k$, we continue the traversal to find the longest prefix run $B_1\cdots B_i$ of blocks from level $\ell$ with $i \le b$. To find the run $B_1\cdots B_i$, we look at the parent of $B_1$ in $H$, which must be a run (possibly, with only one block $B_1$) that ends exactly at the block $B_i$, if $i < b$, or, probably, at a block to the right of $B_b$, if $i = b$. Analogously, the right traversal finds the longest suffix run $B_{j+1}\cdots B_b$ (see details in the definition of $\fin$). Thus, using the notation from the definition of $\fin(\ell, B_1,\ldots,B_b)$, one can say that we have computed, respectively, the prefix and suffix surrounding the recursive $\fin$ part in the sequence $\langle\id(B_i),i\rangle, \fin(2k + 1, B'_1,\ldots, B'_{b'}), \langle\id(B_{j+1}),b-j\rangle$, where $B'_1,\ldots,B'_{b'}$ are respective blocks from level $\ell + 1$ as in the definition of $\fin$ (our algorithm has not yet found these blocks at this point). If $b' = 0$, we abort both traversals as we have constructed the whole fingerprint. If $b' > 0$, we assign $\ell = 2k + 1$ and recursively proceed to the next step of our algorithm computing $\fin(2k + 1, B'_1,\ldots, B'_{b'})$. Note that, if $|B_i| > 2^k$, then we have $i = 0$ (i.e.,  the prefix surrounding the recursive $\fin$ part is empty) and, in the recursive next step, the first block from level $2k + 1$ in the algorithm is just the parent of $B_i$ in $H$ that has the same identifier $\id(B_i)$; the case $|B_b| > 2^k$ is symmetrical.
	
	Case $\ell =2k + 1$, for some $k \ge 0$. Recall that we have the level-$\ell$ blocks $B_1$ and $B_b$ and the goal is to compute $\fin(\ell, B_1,\ldots,B_b)$. As in the definition of $\fin$, we are to compute the prefix and suffix surrounding the recursive $\fin$ part in the sequence $\id(B_1),\ldots,\id(B_i),\fin(2k + 2, B'_1,\ldots, B'_{b'}), \id(B_{j+1}),\ldots,\id(B_b)$, where $i$, $j$ and $B'_1,\ldots,B'_{b'}$ are respective indices and blocks as in the definition of $\fin$ (our algorithm has not yet found these blocks and has not yet found $i$ and $j$ at this point). Let us describe how one can calculate the number $j$; the calculation of $i$ is mostly symmetrical but less nuanced. 
	We perform the right traversal retrieving consecutively the blocks $B_b, B_{b-1},\ldots$ from level $\ell$ until we reach a block $B_j$ with $j \in [1\dd b]$ such that either $|B_j| > 2^k$, or $\infty > \id''(B_{j-2}) > \id''(B_{j-1})$ and $\id''(B_{j-1}) < \id''(B_j) < \infty$ (with $j \ge 3$), where the function $\id''$ is computed isolatedly for the sequence $B_1,\ldots,B_b$. (See details in the definition of $\fin$ in Section~\ref{sec:search}.) If such block $B_j$ is never reached, we simply report $\id(B_1),\ldots,\id(B_b)$ as the fingerprint $\fin(\ell, B_1,\ldots,B_b)$ and abort both traversals as we have the result. Note that we have $\id''(B_1) = \id''(B_2) = \infty$ and, for $h \in (2\dd b]$, to calculate $\id''(B_h)$, we have to know $\id(B_{h-2})$ and $\id(B_{h-1})$. Hence, when we reach the required block $B_j$ and it happens that $|B_j| \le 2^k$, then we had to retrieve also the blocks $B_{j-1},B_{j-2},B_{j-3},B_{j-4}$. Once $i$ and $j$ are calculated, we can assign $\ell = 2k + 2$ and recursively proceed to the next step of our algorithm computing $\fin(2k + 2, B'_1,\ldots, B'_{b'})$ (provided it is not the last step, i.e., $i < j$). There is a subtle caveat here: the first block from the next level $2k + 2$ that should be access on the next step by the right traversal must be the parent of the block $B_j$ and, thus, we have to ``rollback'' some steps of the right traversal that reached the blocks $B_{j-1},B_{j-2},B_{j-3},B_{j-4}$, continuing the traversal from $B_j$.
	
	\textbf{Properties of the traversal trees.}
	Consider the right DFS traversal tree $T_R$; Fig.~\ref{fig:traversal}. (Note that we do not exclude from $T_R$ the ``rolled back'' blocks.)  Denote by $\ell_{\max}$ the largest level of a block in the reported fingerprint. We have $\ell_{\max} = O(\log m)$. Since the size of the fingerprint is $O(\log m\cdot\log\log n)$ and all leaves in $T_R$ have levels at most $\ell_{\max}$, the number of leaves in $T_R$ is $O(\log m\cdot\log\log n)$. Let $L_1,L_2,\ldots,L_z$ be the leaves of $T_R$ listed from left to right. It follows from the algorithm that the level of $L_h$ is non-increasing when $h$ grows from $1$ to $z$ and, for any two leaves $L_h$ and $L_{h+1}$ with $h \in [1\dd z)$, if the leaf $L_h$ is from level $\ell$, then the leaf $L_{h+1}$ has an ancestor $L'$ in $T_R$ from level $\ell$ ($L' = L_{h+1}$ if $L_{h+1}$ itself is from level $\ell$) such that $\send(L_h) + 1 = \sbeg(L')$. It follows from this that, for each $\ell$, if $B'_1,B'_2,\ldots,B'_{z'}$ are all level-$\ell$ blocks in $T_R$ and $\sbeg(B'_1) < \sbeg(B'_2) < \cdots$, then, for each $h \in [1\dd z')$, we have $\send(B'_h) + 1 = \sbeg(B'_{h+1})$; in other words, the tree $T_R$ does not have ``gaps'' inside its block structure. Analogous observations for the left DFS traversal tree $T_L$ are symmetrical. 
	
	\begin{figure}
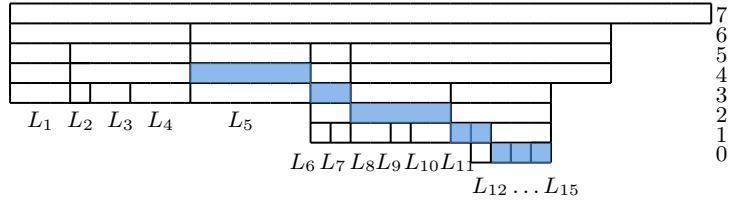

		\begin{center}
			\tikzset{every picture/.style={line width=0.75pt}} 
			

			
		\end{center}
	
		\caption{A schematic depiction of blocks traversed in the tree $T_R$. The blue blocks are reported to the fingerprint (the left part of the fingerprint reported during the traversal of the tree $T_L$ is not depicted). The reported blocks on even levels (0, 2, 4, 6) form runs, sometime one-block runs. The levels of the leaves ($L_1,L_2,\ldots$) are non-increasing. Note that the algorithm sometimes visits blocks below the reported blue blocks: on a level $\ell$, it might visit up to five unreported additional blocks to the left of a reported blue block (it is done in order to find the local minima of $\id''$ that identify the boundary of the reported block). So, in total, we have $O(\log m\cdot\log\log n)$ leaves.}
		\label{fig:traversal}
	\end{figure}
	
	One can view the simultaneous traversal algorithm as if it computes the leaves of $T_R$ from right to left and the leaves of $T_L$ from left to right until they converge at one ``point'' on level $\ell_{\max}$ (with a slight ``overlap'', possibly, and here the notion of ``convergence'' is not precise) so that the trees $T_R$ and $T_L$ ``touch'' each other, in a way, on levels above $\ell_{\max}$: namely, for any $\ell > \ell_{\max}$, if $B_R$ is the leftmost level-$\ell$ block from $T_R$ and $B_L$ is the rightmost level-$\ell$ block from $T_L$, then $\send(B_L) + 1 \ge \sbeg(B_R)$. Note also that, for any $\ell$, if $B_R$ is the rightmost level-$\ell$ block from $T_R$ and $B_L$ is the leftmost level-$\ell$ block from $T_L$, then $\sbeg(B_L) \le \sbeg(B_R)$ (i.e., the tree $T_L$ is located to the ``left'' of $T_R$ in the hierarchy of blocks $H$, with possible overlaps).
	
	\textbf{Emulated algorithm.}
	Now let us transform the simultaneous traversal algorithm for the case when the hierarchy of blocks $H$ is not directly accessible but it is emulated through the jiggly block tree $J$ using Lemma~\ref{lem:traversal}. The simplest emulation would just descend in the trees $T_R$ and $T_L$ computing the children of blocks using Lemma~\ref{lem:traversal}. But it raises the following question: if we first encounter a block $B'$ during the emulated traversal of the tree $T_R$ and the level of the parent of $B'$ in $T_R$ was known to be $\ell$, then, surely, the level of $B'$ is $\ell - 1$ but what is the level in $T_R$ of the children of $B'$ in the hierarchy $H$ obtained using Lemma~\ref{lem:traversal} from the tree $J$? Recall that $B'$ stores in $J$ the lowest level $\ell'$ on which a block with identifier $\id(B')$ appeared in the original hierarchy $H$, denote it $\ell' = \ell(B')$; then, it follows from Lemma~\ref{lem:uniformity} that the obtained children of $B'$ must have level $\ell' - 1$ in $T_R$ and, thus, $B'$ ``induces'' a chain of $\ell - \ell' + 1$ blocks in $T_R$, each of which is a copy of $B'$. We can easily navigate to any block in this chain in $O(1)$ time. 
	
	We slightly speed up the emulation further: when we descend along the rightmost path of $T_R$ to the rightmost leaf of $T_R$ (which is the one-letter block $s[q]$) at the beginning of the algorithm, we encounter sequences of descends of the form $B'_1, \cdots, B'_{z'}$, where all $B'_1, \cdots, B'_{z'}$ are ancestors of $s[q]$ and, for each $h \in [1\dd z')$, $B'_{h+1}$ is the leftmost child of $B'_h$; all this chain can be skipped in $O(\log\log n)$ time using the weighted ancestor data structure on an auxiliary tree $A_L$ that was maintained for $J$ so that we arrive to the block $B'_{z'}$ in one ``hop''. Since, for each $h \in [1\dd z')$, we have $q < \send(B'_h)$, the algorithm will never need these skipped blocks, so we can ignore them in the algorithm. But we need also the levels $\ell(B'_{z'})$ and $\ell(B'_{z'-1})$ for the correct emulation of skipped levels; we can access $B'_{z-1}$ with one more weighted ancestor query in a straightforward way. Symmetrically, for the tree $T_L$, we skip long sequences of descends into rightmost children along the leftmost path of the tree $T_L$. It turns out that this only non-trivial optimization suffices to make the algorithm fast enough. Let us analyse the optimized emulated algorithm.
	
	\textbf{Analysis.}
	Denote by $\hat{T}$ the tree that is the union of the trees $T_R$ and $T_L$ (i.e., the set of blocks in $\hat{T}$ is the union of all blocks from $T_R$ and $T_L$). It is convenient to imaging that the emulation traverses the tree $\hat{T}$, not $T_R$ and $T_L$ separately. $\hat{T}$ has $O(\log m\cdot\log\log n)$ leaves and, hence, $O(\log m\cdot\log\log n)$ branching nodes. Due to Lemma~\ref{lem:traversal}, given a block $B'$ from the tree $J$, one can compute all children of $B'$ from the hierarchy $H$ in $O(\log\log n)$ time. Hence, the computation on all branching nodes in $\hat{T}$ takes $O(\log m\cdot(\log\log n)^2)$ time in total. But the tree $\hat{T}$ might also contain long chains of non-branching nodes. In total, we have visited $O(\log m\cdot\log\log n)$ blocks of interest in $\hat{T}$ during the execution of the algorithm; it remains to estimate how much time was spent on blocks from non-branching chains in $\hat{T}$. 
	
	If we encounter a chain that consists of copies of one block, then we reach any block of interest on such chain in $O(1)$ time, as described above. Thus, we spend on such chains $O(\log m\cdot\log\log n)$ time in total, over all chains (since this is the number of blocks of interest).
	
	Consider a block $B'$ in $\hat{T}$ that has the only child $B''$ in $\hat{T}$ and $\id(B') \ne \id(B'')$. Denote by $\ell$ the level of $B''$. Suppose that $B''$ belongs to $T_R$; the case when $B''$ belongs to $T_L$ is symmetric. Since the tree $T_R$ has no ``gaps'' on level $\ell$,  $B''$ must be either the leftmost or rightmost level-$\ell$ block in $T_R$. Therefore, for $\ell \le \ell_{\max}$, there are in total at most $2\ell_{\max} = O(\log m)$ such blocks $B''$ in $T_R$ and we spend $O(\log m)$ time on them in total. Let $\ell > \ell_{\max}$. 
	Since $B$ is the lowest block in $H$ such that $\sbeg(B) \le p < q \le \send(B)$, the block $B''$ cannot be the only level-$\ell$ block in $\hat{T}$. Recall that $B''$ is the only child of $B'$ in $\hat{T}$. Therefore, since the tree $T_L$ has no ``gaps'' on level $\ell$ and the trees $T_L$ and $T_R$ ``touch'' each other (in a precise sense described above) on level $\ell > \ell_{\max}$, then $B''$ is either the leftmost or rightmost child of $B'$. In former (respectively, latter) case, $B''$ belongs to a part of the rightmost (respectively, leftmost) path along $T_R$ (respectively, $T_L$) that was skipped using the weighted ancestor data structure and $B''$ was not important for the algorithm.
\end{proof}
\fi


\begin{figure}[h]
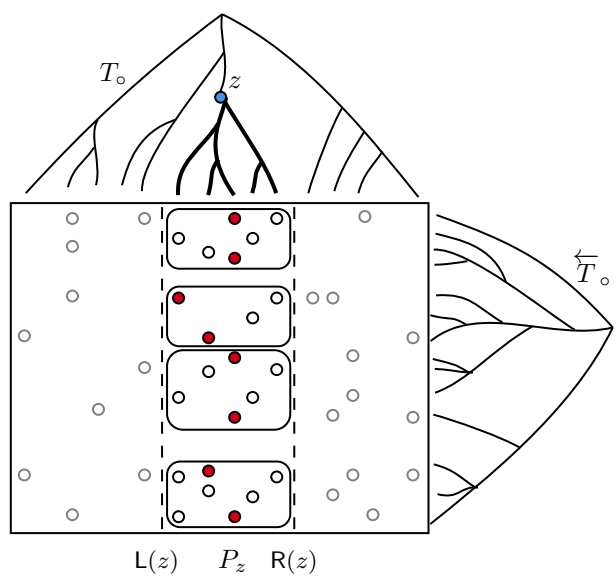

	\begin{center}
		
		\tikzset{every picture/.style={line width=0.75pt}}
		

		
	\end{center}
	
	\vspace{-0.2cm}
	\caption{A schematic depiction of $T_\circ$ and $\lvec{T}_\circ$ with the points $P$. The subset $P_z \subseteq P$ corresponding to a node $z\in T_\circ$ is depicted as split into chunks $P_z^a$. In each chunk $P_z^a$, its points with maximum and minimum $y$-coordinate are red. $V_z$ stores the values $\mathsf{v}(y)$ for $y$-coordinates of exactly these red points ($O(|P_z| / \log\log n)$ values in total). $Q_z$ stores, for each chunk $P_z^a$, the minimum of $p_{x,y}$ for all $(x,y) \in P_z^a$. The elements of each chunk $P_z^a$ are not stored explicitly but can be retrieved in $O(\log n + |P_z^a|\frac{\log n}{\log\log n}) = O(\log n)$ time using one range reporting on $P$.}
	\label{fig:chunks}
\end{figure}

\end{document}